\documentclass{ijuc} 

\usepackage{graphicx,amssymb, amstext, amsfonts}
\usepackage{comment} 

\newtheorem{theorem}{Theorem}[section]

\newtheorem{definition}[theorem]{Definition}

\newenvironment{proof}[1]{\begin{trivlist}\item[\hskip\labelsep{\it
Proof.\ }]}{\hspace*{\fill} {$\Box$}\end{trivlist}}

\makeatletter

\newcommand{\Rmnum}[1]{\expandafter\@slowromancap\romannumeral #1@}
\makeatother

\usepackage{multicol}
\title{Program-size versus Time complexity\\{\small Slowdown and speed-up phenomena in the micro-cosmos\\ 
\vspace{-0.2cm} 
of small Turing machines}}
\author{Joost J. Joosten\inst{1} \and Fernando Soler-Toscano\inst{2} \and Hector Zenil\inst{3,4}}
\institute{Departament de l\`ogica, hist\`oria i filosofia de la ci\`encia, Universitat de Barcelona \email{jjoosten@ub.edu}\and
Grupo de L\'ogica, Lenguaje e Informaci\'on,
Departamento de Filosof\'{\i}a, L\'ogica, y Filosof\'{\i}a de la Ciencia,
Universidad de Sevilla
\email{fsoler@us.es}
 \and 
 Laboratoire d'Informatique Fondamentale de Lille\\(CNRS), Universit\'e de Lille I
 \and
 Wolfram Research, Inc.
\email{hectorz@wolfram.com}
 }
\begin{document}

\maketitle

\begin{abstract}
The aim of this paper is to undertake an experimental investigation of the trade-offs between program-size and time computational complexity. The investigation includes an exhaustive exploration and systematic study of the functions computed by the set of all 2-color Turing machines with 2 states --we will write (2,2)-- and 3 states --we write (3,2)-- with particular attention to the runtimes, space-usages and patterns corresponding to the computed functions when the machines have access to larger resources (more  states). 

We report that the average runtime of Turing machines computing a function increases --with a probability close to one-- with the number of states, indicating that machines not terminating (almost) immediately tend to occupy all the resources at hand. We calculated all time complexity classes to which the algorithms computing the functions found in both (2,2) and (3,2) belong to, and made comparison among these classes. For a selection of functions the comparison is extended to (4,2).

Our study revealed various structures in the micro-cosmos of small Turing Machines. Most notably we observed ``phase-transitions'' in the halting-probability distribution. Moreover, it is observed that small initial segments fully define a function computed by a TM.

\noindent \textbf{Keywords}: small Turing machines, Program-size complexity, Kolmogorov-Chaitin complexity, space/time complexity, computational complexity, algorithmic complexity.
\end{abstract}

\section{Introduction}

Among the several measures of computational complexity there are measures focusing on the minimal description of a program and others quantifying the resources (space, time, energy) used by a computation.
%
%
This paper is a reflection of an ongoing project with the ultimate goal of contributing to the understanding  of relationships between various measures of complexity by means of computational experiments. In particular in the current paper we did the following.

We focused on small Turing Machines and looked at the kind of functions that are computable on them focussing on the runtimes. We then study how allowing more computational resources in the form of Turing machine states affect the runtimes of TMs computing these functions. We shall see that in general and on average, more resources leads to slower computations. In this introduction we shall briefly introduce the main concepts central to the paper.

\subsection{Two measures of complexity}

The long run aim of the project focuses on the relationship between various complexity measures, particularly descriptional and computational complexity measures. In this subsection we shall briefly and informally introduce them. 

In the literature there are results known to theoretically link some complexity notions. 
For example, in \cite{calude}, runtime probabilities were estimated based on Chaitin's heuristic principle as formulated in \cite{chaitin}. Chaitin's principle is of descriptive theoretic nature and states that \emph{the theorems of a finitely-specified theory cannot be significantly more complex than the theory itself}.

Bennett's concept of logical depth combines the concept of time complexity and program-size complexity \cite{bennett,bennett2} by means of the time that a decompression algorithm takes to decompress an object from its shortest description.

Recent work by Neary and Woods \cite{neary} has shown that the simulation of cyclic tag systems by cellular automata is effected with a polynomial slow-down, setting a very low threshold of possible non-polynomial tradeoffs between program-size and computational time complexity.

\subsubsection{Computational Complexity}
Computational complexity \cite{cook2,levin} analyzes the difficulty of computational problems in terms of computational resources. The computational time complexity of a problem is the number of steps that it takes to solve an instance of the problem using the most efficient algorithm, as a function of the size of the representation of this instance. 

As widely known, the main open problem with regard to this measure of complexity is the question of whether problems that can be solved in non-deterministic polynomial time can be solved in deterministic polynomial time, aka the P versus NP problem. Since P is a subset of NP, the question is whether NP is contained in P. If it is, the problem may be translated as, for every Turing machine computing an NP function 
 there is (possibly) another Turing machine that does so in P time. In principle one may think that if in a space of all Turing machines with a certain fixed size there is no such a P time machine for the given function 
 (and because a space of smaller Turing machines is always contained in the larger) only by adding more resources a more efficient algorithm, perhaps in P, might be found. We shall see that adding more resources almost certainly yields to slow-down.

\subsubsection{Descriptional Complexity}
The algorithmic or program-size complexity \cite{kolmogorov,chaitin} of a binary string is informally defined as the shortest program that can produce the string. There is no algorithmic way of finding the shortest algorithm that outputs a given string

More precisely, the complexity of a bit string $s$ is the length of the string's shortest program in binary on a fixed universal Turing machine. 
A string is said to be complex or random if its shortest description cannot be much more shorter than the length of the string itself. And it is said to be simple if it can be highly compressed. There are several related variants of algorithmic complexity or algorithmic information.

In terms of Turing machines, if $M$ is a Turing machine which on input $i$ outputs string $s$, then the concatenated string $\langle M,i\rangle$ is a description of $s$. The size of a Turing machine in terms of the number of states (s) and colors (k) (aka known as symbols) 
can be represented 
by the product $s \cdot k$. Since we are fixing the number of colors to $k=2$ in our study, we increase the number of states $s$ as a mean for increasing the program-size (descriptional) complexity of the Turing machines in order to study any possible tradeoffs with any of the other complexity measures in question, particularly computational (time) complexity.

\subsection{Turing machines}
Throughout this project the computational model that we use will be that of Turing
machines. Turing machines are well-known models for universal
computation. This means, that anything that can be computed at all,
can be computed on a Turing machine.


In its simplest form, a Turing machine consists of a two-way infinite
tape that is divided in adjacent cells. Each cell can be either blank
or contain a non-blank color (symbol). The Turing machine comes with a
``head" that can move over the cells of the tape. Moreover, the
machine can be in 
various states. At each step in time, the machine
reads what color is under the head, and then, depending on in what
state it is writes a (possibly) new color in the cell under the head,
goes to a (possibly) new state and have the head move either left or
right. A specific Turing machine is completely determined by its
behavior at these time steps. One often speaks of a transition rule,
or a transition table. Figure~\ref{fig:TMRule2506} depicts graphically
such a transition rule when we only allow for 2 colors, black and white and where there are two states, State 1 and State 2.

\begin{figure}[htb!]
  \centering
  \includegraphics[height=1.3cm]{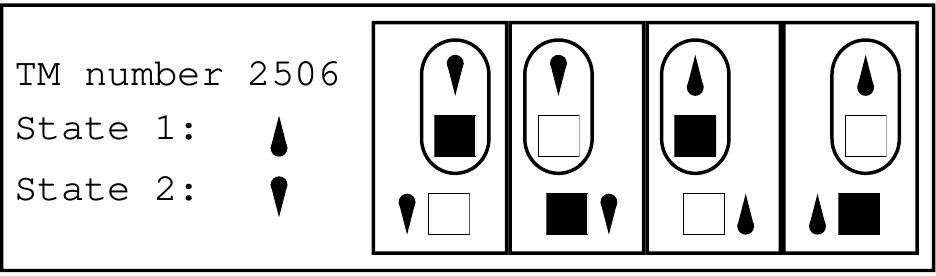}
  \caption{Transition table of a 2-color 2-state Turing machine with Rule 2506 according to Wolfram's enumeration and Wolfram's visual representation style \cite{wolfram}, \cite{JoostenDemonstration2010}.}
  \label{fig:TMRule2506}
\end{figure}

For example, the head of this machine will only move to the right,
write a black color and go to State 2 whenever the machine was in
State 2 and it read a blank symbol.

We shall often refer to the collection of TMs with $k$ colors and $s$ states as a TM space.
From now on, we shall write (2,2) for the space of TMs with 2 states and 2 colors, and (3,2) for the space of TMs with 3 states and 2 colors, etc.

%

\subsection{Relating notions of complexity}

%

We relate and explore throughout the experiment the connections between descriptional complexity and time computational complexity. One way to increase the descriptional complexity of a Turing machine is enlarging its transition table description by adding a new state. So what we will do is, look at time needed to perform certain computational tasks first with only 2 states, and next with 3 and 4 states.

Our current findings suggest that even if a more efficient Turing machine algorithm solving a problem instance may exist, the probability of picking a machine algorithm at random among the TMs that solve the problem in a faster time has probability close to 0 because the number of slower Turing machines computing a function outnumbers the number of possible Turing machines speeding it up by a fast growing function. 




\subsection{Investigating the micro-cosmos of small Turing machines}
We know that small programs are capable of great complexity. For example, computational universality occurs in cellular automata with just 2 colors and nearest neighborhood (Rule 110, see \cite{wolfram,cook}) and also (weak) universality in Turing machines with only 2-states and 3-colors \cite{wolfram23}.

For all practical purposes one is 
restricted to perform experiments with small Turing machines (TMs) if one pursuits a thorough investigation of complete spaces for a certain size.
Yet the space of these machines is rich and large enough to allow for interesting and insightful comparison, draw some preliminary conclusions and shed light on the relations between measures of complexity. 

As mentioned before, in this paper, we look at TMs with 2 states and 2 colors and compare them to TMs more states. The main focus is on the functions they compute and the runtimes for these functions
However, along our investigation we shall deviate from time to time from our main focus and marvel at the rich structures present in what we like to refer to as \emph{the micro-cosmos of small Turing machines}.
Like, what kind of, and how many functions are computed in each space? What kind of runtimes and space-usage do we typically see and how are they arranged over the TM space? What are the sets that are definable using small Turing machines? How many input values does one need to fully determine the function computed by a TM? We find it amazing how rich the encountered structures are even when we use so few resources.
%
%

\subsection{Plan of the paper}
After having introduced the main concepts of this paper and after having set out the context in this section, the remainder of this paper is organized as follows.
In Section \ref{Section:Methodology} we will in full detail describe the experiment, its methodology and the choices that were made leading us to the current methodology. In Section \ref{Section:22} we present the structures that we found in (2,2). The main focus is on runtimes but a lot of other rich structures are exposed there. In Section \ref{Section:32} we do the same for the space (3,2). Section \ref{Section:42} deals with (4,2) but does not disclose any additional structure of that space as we did not exhaustively search this space. Rather we sampled from this space looking for functions we selected from (3,2).
In Section \ref{Section:Comparing} we compare the various TM spaces focussing on the runtimes of TMs that compute a particular function. 
\section{Methodology and description of the experiment}\label{Section:Methodology}

In this section we shall briefly restate the set-up of our experiment to then fill out the details and motivate our choices. We try to be as detailed as possible for a readable paper. For additional information, source code, figures and obtained data can be requested from any of the authors.

\subsection{Methodology in short}
It is not hard to see that any computation in (2,2) is also present in (3,2).
At first, we look at TMs in (2,2) and compare them to TMs in (3,2). In particular we shall study the functions they compute
and the time they take to compute in each space. 

The way we proceeded is as follows. We ran all the TMs in (2,2) and (3,2) for 1000 steps for the first 21 input values $0,1, \ldots,20$. If a TM does not halt by 1000 steps we simply say that it diverges. We saw that certain TMs defined a regular progression of runtimes that needed more than 1000 steps to complete the calculation for larger input values. For these regular progressions we filled out the values manually as described in Subsection \ref{Section:Cleansing}. Thus, we collect all the functions on the domain $[0,20]$ computed in (2,2) and (3,2) and investigate and compare them in terms of run-time, complexity  and space-usage. 
We selected some interesting functions from (2,2) and (3,2). For these functions we searched by sampling for TMs in (4,2) that compute them so that we could include (4,2) in our comparison. 

Clearly, at the outset of this project we needed to decide on at least the following issues:
\begin{enumerate}
\item \label{item:representingNumbers}
How to represent numbers on a TM?

\item
How to decide which function is computed by a particular TM.

\item
Decide when a computation is considered finished.
\end{enumerate}

The next subsections will fill out the details of the technical choices made and provide motivations for these choices. Our set-up is reminiscent of and motivated by a similar investigation in Wolfram's book \cite{wolfram}, Chapter 12, Section 8.

\subsection{Resources}\label{subsection:Resources}
There are ${(2sk)}^{sk}$ s-state k-color Turing machines. That means 4\,096 in (2,2) and 2\,985\,984 TMs in (3,2). In short, the number of TMs grows exponentially in the amount of resources. Thus, in representing our data and conventions we should be as economical as possible in using our resources so that exhaustive search in the smaller spaces still remains feasible. For example, an additional halting state will immediately increase the search space\footnote{Although in this case not exponentially so, as halting states define no transitions.}.

\subsection{One-sided Turing Machines}

In our experiment we have chosen to work with one-sided TMs. That is
to say, we work with TMs with a tape that is unlimited to the left but
limited to the right-hand side.  One sided TMs are a common convention in the literature just perhaps 
slightly less
common 
than the 
two sided convention. The following considerations led
us to work with one-sided TMs.

\begin{itemize}
\item[-] Efficient (that is, non-unary) number representations are
  place sensitive. That is to say, the interpretation of a digit
  depends on the position where the digit is in the number. Like in
  the decimal number 121, the leftmost 1 corresponds to the
  centenaries, the 2 to the decades and the rightmost 1 to the
  units. On a one-sided tape which is unlimited to the left, but
  limited on the right, it is straight-forward how to interpret a tape
  content that is almost everywhere zero. For example, the tape
  $\ldots 00101$ could be interpreted as a binary string giving rise
  to the decimal number 5. For a two-sided infinite tape one can think
  of ways to come to a number notation, but all seem rather arbitrary.

\item[-] With a one-sided tape there is no need for an extra halting
  state. We say that a computation simply halts whenever the head
  ``drops off'' the tape from the right hand side. That is, when the
  head is on the extremal cell on the right hand side and receives the
  instruction to moves right. A two-way unbounded tape would require an
  extra halting state which, in the light of considerations in
  \ref{subsection:Resources} is undesirable.
\end{itemize}
On the basis of these considerations, and the fact that some work has been done before in the lines of this experiment \cite{wolfram} that also contributed to motivate our own investigation, we decided to fix the TM formalism and choose the one-way tape model.

\subsection{Unary input representation}
Once we had chosen to work with TMs with a one-way infinite tape, the next choice is how to represent the input values of the function. When working with two colors, there are basically two choices to be made: unary or binary. However, there is a very subtle point if the input is represented in binary. If we choose for a binary representation of the input, the class of functions that can be computed is rather unnatural and very limited. 

The main reason is as follows. Suppose that a TM on input $x$ performs
some computation. Then the TM will perform the very same computation
for any input that is the same as $x$ on all the cells that were
visited by the computation. That is, the computation will be the same
for an infinitude of other inputs thus limiting the class of functions
very severely. 
We can make this more precise in Theorem \ref{theorem:Strips} below.
The theorem shows that coding the input in $k$-ary where $k$ is the number of alphabet symbols, severely restricts the class of computable functions in $(s,k)$ . 
For convenience and for the sake of our presentation, we only consider the binary case, that is, $k=2$. 
\begin{definition}
A subset of the natural numbers is called a \emph{strip} if it is of the form $\{  a + b\cdot n \mid n\in \omega \}$ for certain fixed natural numbers $a$ and $b$.
\end{definition}

\begin{definition}
A \emph{strip} of a function $f$ is simply $f$ restricted to some subdomain $\mathcal{D}$ where $\mathcal{D}$ is a  subset of the natural numbers that is a strip. 
\end{definition}

\begin{theorem}[The Strips Theorem]\label{theorem:Strips}
Let $f$ be a function that is calculated by a one sided TM. Then $f$ consists of strips of functions of the form $a + x$.
\end{theorem}

\begin{proof}
SSuppose $f$ halts on input $i$. Let $n$ be the left-most cell visited by the TM on input $i$. Clearly, changing the input on the left-hand side of this $n$-th cell will not alter the calculation. In ohter words, the calculation will be the same for all inputs of the form $x\cdot 2^{n+1} +i$. What will be the output for these respective inputs. Well, let $s_i = f(i)$, then the outputs for these infinitely many inputs will consist of this output $s_i$ together with the part of the tape that was unaltered by the computation. Thus, $f(x\cdot 2^{n+1} +i) = x\cdot 2^{n+1} + f(i) = x\cdot 2^{n+1} + i + (f(i)-i)$: behold our strip of the form $a+x$.   
\end{proof}

We can see the smallest elements that are defining a strip, the $i$ in the proof above, as sort of prime elements. The first calculation on a TM defines a strip. The next calculation on an input not already in that strip defines a new strip and so forth. Thus, the progressively defined strips define new prime elements and the way they do that is quite similar to Eratosthenes' Sieve. For various sieve-generated sets of numbers it is known that they tend to be distributed like the primes (\cite{wunderlich}) in that the prime elements will vanish and only occur with probability $\frac{1}{\log (x)}$.  If this would hold for the smallest elements of our strips too, in the limit, each element would belong with probability one to some previously defined strip. And each prime element defines some function of the form $x + a_i$. The contribution of the $a_i$ vanishes in the limit so that we end up with the identity function. In this sense, each function calculated by a one-sided TM would calculate the identity in the limit.

The Strips Theorem is bad for two reasons. Firstly, it shows that the class of computable functions is severely restricted. We even doubt that universal function can occur within this class of functions. And, if universal functions do occur, at the cost of how much coding is that the case. In other words, if possible at all, how strong should the coding mechanism be that is needed to represent a computable problem within the strips functions. Secondly, there is the problem of incomparability of functions computed in TM spaces. It is easily seen that Theorem \ref{theorem:Strips} generalizes to more colors $k$. Still, the strips will compute functions of the form $a + x$, however the strips themselves will be of the form $x\cdot k^{n+1} + i$. Thus, if at some stage the current project were to be extended and one wishes to study the functions that occur in spaces that use a different number of colors, by the Strips Theorem, this intersection is expected to be very small.

The following theorem tells us that the restriction on the class of computable functions when using $k$-ary input representation has nothing to do with the fact that the computation was done on a single sided TM and the same phenomena occurs in double-sided TMs.

\begin{theorem}
Let $f$ be a function that is calculated by a two sided TM. Then $f$ consists of strips of functions of the form $a + 2^l \cdot x$ with $l\in \mathbb{Z}$.
\end{theorem}

\begin{proof}
 xAs before the proof is based on the observation that altering the input on that part of the tape that was never visited will not influence the calculation. The only thing to be taken into account now is that the output can be shifted to the right (depending on conventions). So that in the end you see that the function, with $2^m$ units to the right correspond to $2^n$ units up, hence a tangent of $2^l$ for some $l\in \mathbb{Z}$.  
\end{proof}

On the basis of these considerations we decided to represent the input in unary. Moreover, from a theoretical viewpoint it is desirable to have the empty tape input different from the input zero, thus the final choice for our input representation is to represent the number $x$ by $x+1$ consecutive 1's.\\ 
\medskip

The way of representing the input in unary has two serious draw-backs:

\begin{enumerate}
\item 
The input is very homogeneous. Thus, it can be the case that TMs that expose otherwise very rich and interesting behavior, do not do so when the input consists of a consecutive block of 1's.

\item
The input is lengthy so that runtimes can grow seriously out of hand. See also our remarks on the cleansing process below.
\end{enumerate}

We mitigate these objections with the following considerations.
\begin{enumerate}
\item 
Still interesting examples are found. And actually a simple informal argument using the Church-Turing thesis shows that universal functions will live in a canonical way among the thus defined functions.
\item
The second objection is more practical and more severe. However, as the input representation is so homogeneous, often the runtime sequences exhibit so much regularity that missing values that are too large can be guessed. We shall do so as described in Subsection \ref{Section:Cleansing}.
\end{enumerate}

\subsection{Binary output convention} 
None of the considerations for the input conventions applies to the output convention. Thus, it is wise to adhere to an output convention that reflects as much information about the final tape-configuration as possible. Clearly, by interpreting the output as a binary string, from the output value the output tape configuration can be reconstructed. Hence, our outputs, if interpreted, will be so as binary numbers.

\begin{definition}[Tape Identity]\label{definition:TapeIdentity}
We say that a TM computes the \emph{tape identity} when the tape configuration at the end of a computation is identical to the tape configuration at the start of the computation.
\end{definition}

%
The output representation can be seen as a simple operation between systems, taking one representation to another. 
The main issue is, how does one keep the structure of a system when represented in another system, such that, moreover, no additional essential complexity is introduced.

For the tape identity, for example, one may think of representations that, when translated from one to another system, preserve the simplicity of the function. However, a unary input convention and a binary output representation immediately endows the tape identity with an exponential growth rate. In principle this need not be a problem. However, computations that are very close to the tape identity will give rise to number theoretic functions that are seemingly very complex.
However, as we shall see, in our current set-up there will be few occasions where we actually do interpret the output as a number other than for representational purposes. In most of the cases the raw tape output will suffice.

\subsection{The Halting Problem and Rice's theorem}
By the Halting Problem and Rice's theorem we know that it is in general undecidable to know wether a function is computed by a particular TM and whether two TMs define the same function. The latter is the problem of extensionality (do two TMs define the same function?) known to be undecidable by Rice's theorem. It can be the case that for TMs of the size considered in this paper, universality is not yet attained\footnote{Recent work (\cite{woodsNeary}) has shown some small two-way infinite tape universal TMs. It is known that there is no universal machine in the space of two-way unbounded tape (2,2) Turing machines but there is known at least one weakly universal Turing machine in (2,3)\cite{wolfram} and it may be (although unlikely) the case that a weakly universal Turing machine in (3,2) exists.}, that the Halting Problem is actually decidable in these small spaces and likewise for extensionallity. 

As to the Halting Problem, we simply say that if a function does not halt after 1000 steps, it diverges. Theory tells that the error thus obtained actually drops exponentially with the size of the computation bound \cite{calude} and we re-affirmed this in our experiments too as is shown in Figure \ref{figure:haltingProbDistr22}. After proceeding this way, we see that certain functions grow rather fast and very regular up to a certain point where they start to diverge. These obviously needed more than 1000 steps to terminate. We decided to complete these obvious non-genuine divergers manually. This process is referred to as \emph{cleansing} and shall be addressed with more detail in the next subsection.  

As to the problem of extensionality, we simply state that two TMs calculate the same function when they compute (after cleansing) the same outputs on the first 21 inputs 0 through 20 with a computation bound of 1000 steps. We found some very interesting observations that support this approach: for the (2,2) space the computable functions are completely determined by their behavior on the first 3 input values 0,1,2. For the $(3,2)$ space the first 8 inputs were found to be sufficient to determine the function entirely.

\subsection{Cleansing the data}\label{Section:Cleansing}

As mentioned before, the Halting problem is undecidable so one will always err when mechanically setting a cut-off value for our computations. The choice that we made in this paper was as follows. We put the cut-off value at 1000. After doing so, we looked at the functions computed. For those functions that saw an initial segment with a very regular progression of runtimes, for example 16, 32, 64, 128, 256, 512, -1, -1, we decided to fill out the the missing values in a mechanized way. It is clear that, although better than just using a cut-off value, we will still not be getting all functions like this. Moreover, there is a probability that errors are made while filling out the missing values. However we deem the error not too significant, as we have a uniform approach in this process of filling out, that is, we apply the same process for all sequences, either in (2,2) or in (4,2) etc. Moreover, we know from theory (\cite{calude}) that most TMs either halt quickly or never halt at all and we affirmed this experimentally in this paper. Thus, whatever error is committed, we know that the effect of it is eventually only marginally. In this subsection we shall describe the way we mechanically filled out the regular progressions that exceeded the computation bound.
%

We wrote a so-called predictor program that was fed incomplete sequences and was to fill out the missing values.
The predictor program is based on the function \texttt{FindSequenceFunction}\footnote{\texttt{FindSequenceFunction} takes a finite sequence of integer values $\{a_1, a_2, \ldots \}$ and retrieves a function that yields the sequence $a_n$. 
It works by finding solutions to difference equations represented by  the expression \texttt{DifferenceRoot} in \emph{Mathematica}. By default, \texttt{DifferenceRoot} uses early elements in the list to find candidate functions, then validates the functions by looking at later elements. \texttt{DifferenceRoot} is generated by functions such as \texttt{Sum}, \texttt{RSolve} and \texttt{SeriesCoefficient}, also defined in \emph{Mathematica}. \texttt{RSolve} can solve linear recurrence equations of any recurring order with constant coefficients. It can also solve many linear equations (up to second  recurring order) with non-constant coefficients, as well as many nonlinear equations. For more information we refer to the extensive online \emph{Mathematica} documentation.} built-in to the computer algebra system \emph{Mathematica}. Basically, it is not essential that we used \texttt{FindSequence\\
Function} or any other intelligent tool for completing sequences as long as the cleansing method for all TM spaces is applied in the same fashion. A thorough study of the cleansing process, its properties, adequacy and limitations is presented in \cite{ZJS2010}.
%
The predictor pseudo-code is as follows:

\begin{enumerate}
\item
Start with the finite sequence of integer values (with -1 values in the places the machine didn't halt for that input index).
\item
Take the first n consecutive non-divergent (convergent) values, where $n \geq 4$ (if there is not at least a segment with 4 consecutive non divergent values then it gives up).
\item
Call \texttt{FindSequenceFunction} with the convergent segment and the first divergent value.
\item
Replace the first divergent value with the value calculated by evaluating the function found by \texttt{FindSequenceFunction} for that sequence position.
\item
If there are no more -1 values stop otherwise trim the sequence to the next divergent value and go to 1.
\end{enumerate}
This is an example of a (partial) completion: Let's assume one has a sequence (2, 4, 8, 16, -1, 64, -1, 257, -1, -1) with 10 values. The predictor returns: (2, 4, 8, 16, 32, 64, 128, 257, -1, -1) because up to 257 the sequence seemed to be $2^n$ but from 257 on it was no longer the case, and the predictor was unable to find a sequence fitting the rest.

The prediction function was constrained by 1 second, meaning that the process stops if, after a second of trying, no prediction is made, leaving the non-convergent value untouched. This is an example of a completed Turing machine output sequence. Given (3, 6, 9, 12, -1, 18, 21, -1, 27, -1, 33, -1) it is retrieved completed as (3, 6, 9, 12, 15, 18, 21, 24, 27, 30, 33, 36). Notice how the divergent values denoted by $-1$ are replaced with values completing the sequence with the predictor algorithm based in \emph{Mathematica's} \texttt{FindSequenceFunction}. 


%
\subsubsection{The prediction vs. the actual outcome}

For a prediction to be called successful we require that the output, runtime and space usage sequences coincide in every value with the step-by-step computation (after verification). One among three outcomes are possible:

\begin{itemize}
\item Both the step-by-step computation and the sequences obtained with \texttt{FindSequenceFunction} completion produce the same data, which leads us to conclude that the prediction was accurate.
\item The step-by-step computation produces a non-convergent value $-1$, meaning that after the time bound the step-by-step computation didn't produce any new convergent value that wasn't also produced by the FindSequeceFunction (which means that either the value to be produced requires a larger time bound, or that the\\
\texttt{FindSequenceFunction} algorithm has failed, predicting a convergent value where it is actually divergent).
\item The step-by-step computation produced a value that the \\
\texttt{FindSequenceFunction} algorithm did not predict.
\end{itemize}

In the end, the predictor indicated what machines we had to run for larger runtimes in order to complete the sequences up to a final time bound of $200\,000$ steps for a subset of machines that couldnÕt be fully completed with the predictor program.
The number of incorrectly predicted (or left incomplete) in (3,2) was 90 out of a total 3368 sequences completed with the predictor program. In addition to these 45 cases of incorrect completions, we found 108 cases where the actual computation produced new convergent values that the predictor could not predict. The completion process led us to only eight final non-completed cases, all with super fast growing values.

In (4,2) things werenÕt too different. Among the $30\,955$ functions that were sampled motivated by the functions computed in (3,2) that were found to have also been computed in (4,2) (having in mind a comparison of time complexity classes) only 71 cases could not be completed by the prediction process, or were differently computed by the step-by-step computation. That is only 0.00229 of the sequences, hence in both cases allowing us to make accurate comparisons with low uncertainty in spite of the Halting Problem and the problem of very large (although rare) halting times.

\subsection{Running the experiment}

To explore the different spaces of TMs we wrote a TM simulator in the programming language C. 
We tested this C language simulator against the \texttt{TuringMachine} function
in \emph{Mathematica} as it used the same encoding for TMs. It was checked and found in concordance for the whole (2,2) space and a sample of the (3,2) space.

  We run the simulator in the cluster of the CICA (Centro de
  Inform\'atica Cient\'{\i}fica de Andaluc\'{\i}a\footnote{Andalusian
    Centre for Scientific Computing: \texttt{http://www.cica.es/}.}). To explore the (2,2)
  space we used only one node of the cluster and it took 25
  minutes. The output was a file of 2 MB. For (3,2) we used 25
  nodes (50 microprocessors) and took a mean of three hours in each
  node. All the output files together fill around 900 MB.

\section{Investigating the space of 2-states, 2-colors Turing machines}\label{Section:22}

In this section we shall have our first glimpse into the fascinating micro-cosmos of small Turing machines. We shall see what kind of computational behavior is found among the functions that live in (2,2) and reveal various complexity-related properties of the (2,2) space.

\begin{definition}
In our context and in the rest of this paper, an \emph{algorithm} computing a function is one particular set of 21 quadruples of the form $$\langle \mbox{input value}, \mbox{output value}, \mbox{runtime}, \mbox{space usage}\rangle$$ for each of the input values $0, 1, \ldots, 20$, where the output, runtime and space-usage correspond to that particular input.
\end{definition}

In the cleansed data of (2,2) we found 74 functions and a total of 138 different algorithms computing them.

\subsection{Determinant initial segments}

An indication of the complexity of the (2,2) space is the number of inputs needed to determine a function. In the case of (2,2) this number of inputs is only 3. For the first input, the input 0, there are 11 different outputs. The following list shows these different outputs (first value in each pair) and the frequency they appear with
(second value in each pair). Output \texttt{-1} represents the divergent one:

\begin{verbatim} 
{{3, 13}, {2, 12}, {-1, 10}, {0, 10}, {1, 10},
 {7, 6}, {6, 4}, {15, 4}, {4, 2}, {5, 2}, {31, 1}}
\end{verbatim}

For two inputs there are 55 different combinations and for three we
find all the 74 functions. The first input is most significant; without it,
the other inputs only appear in 45 different combinations. This is
because there are many functions with different behavior for the
first input than for the rest.

We find it interesting that only 3 values of a TM are needed to fully determine its behavior in the full (2,2) space that consists of 4\,096 different TMs. Just as a matter of analogy we bring the $\mathbf{C}^{\infty}$ functions to mind. These infinitely often differentiable continuous functions are fully determined by the outputs on a countable set of input values. It is an interesting question how the minimal number of input values needed to determine a TM grows relative to the total number of $(2 \cdot s \cdot k)^{s\cdot k}$ many different TMs in $(s,k)$  space, or relative to the number of defined functions in that space.

\subsection{Halting probability} In the cumulative version of Figure~\ref{figure:haltingProbDistr22} we see that more than 63\% of executions stop after 50 steps, and little growth is obtained after more
steps. Considering that there is an amount of TMs that never halt, it
is consistent with the theoretical result in \cite{calude} that most TMs stop
quickly or never halt.

Let us briefly comment on Figure~\ref{figure:haltingProbDistr22}. First of all, we stress that the halting probability ranges over all pairs of TMs in (2,2) and all inputs between 0 and 20. Second, it is good to realize that the graph is some sort of best fit and leaves out zero values in the following sense. It is easy to see that on the one-sided TM halting can only occur after an odd number of steps. Thus actually, the halting probability of every even number of steps is zero. This is not so reflected in the graph because of a smooth-fit.

We find it interesting that Figure~\ref{figure:haltingProbDistr22} shows features reminiscent of phase transitions. Completely contrary to what we would have expected, these ``phase transitions" were even more pronounced in $(3,2)$ as one can see in Figure \ref{fig:runProb32}.

\begin{figure}
\begin{center}
  \includegraphics[width=10.7cm]{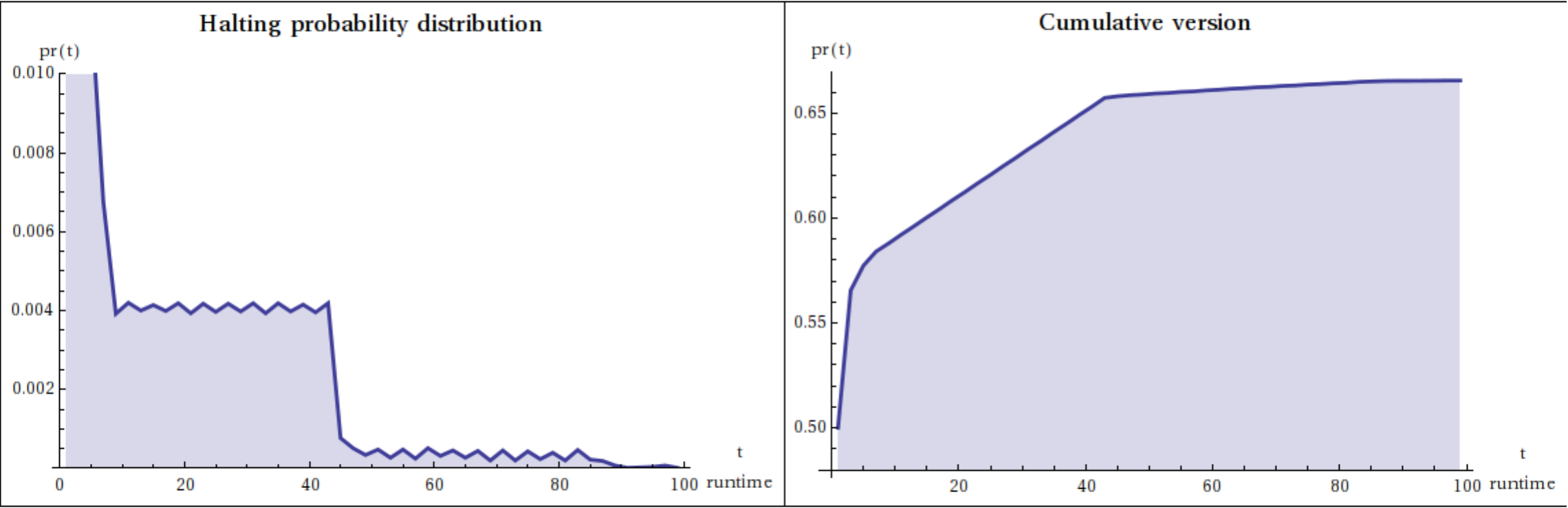}
  \caption{Halting times in (2,2).}\label{figure:haltingProbDistr22}
\end{center}
\end{figure}

\subsection{Phase transitions in the halting probability distribution}\label{Section:PhaseTransitions}

Let consider Figure \ref{figure:haltingProbDistr22} again. Note that in this figure only pairs of TMs and inputs are considered that halt in at most 100 steps.
The probability of stopping (a random TM in $(2,2)$ with a random input in
0 to 20) in at most 100 steps is $0.666$. The probability of stopping
in any number of steps is $0.667$, so most TMs stop quickly of do not
stop.

We clearly observe a phase transition phenomenon. To investigate the cause of this, let us
consider the set of runtimes and the number of their
occurrences. Figure~\ref{fig:occrun} shows at the left the 50 smallest
runtimes and the number of occurrences in the space that we have
explored.
\begin{figure}[htbp!]
  \centering
  \includegraphics[width=10.7cm]{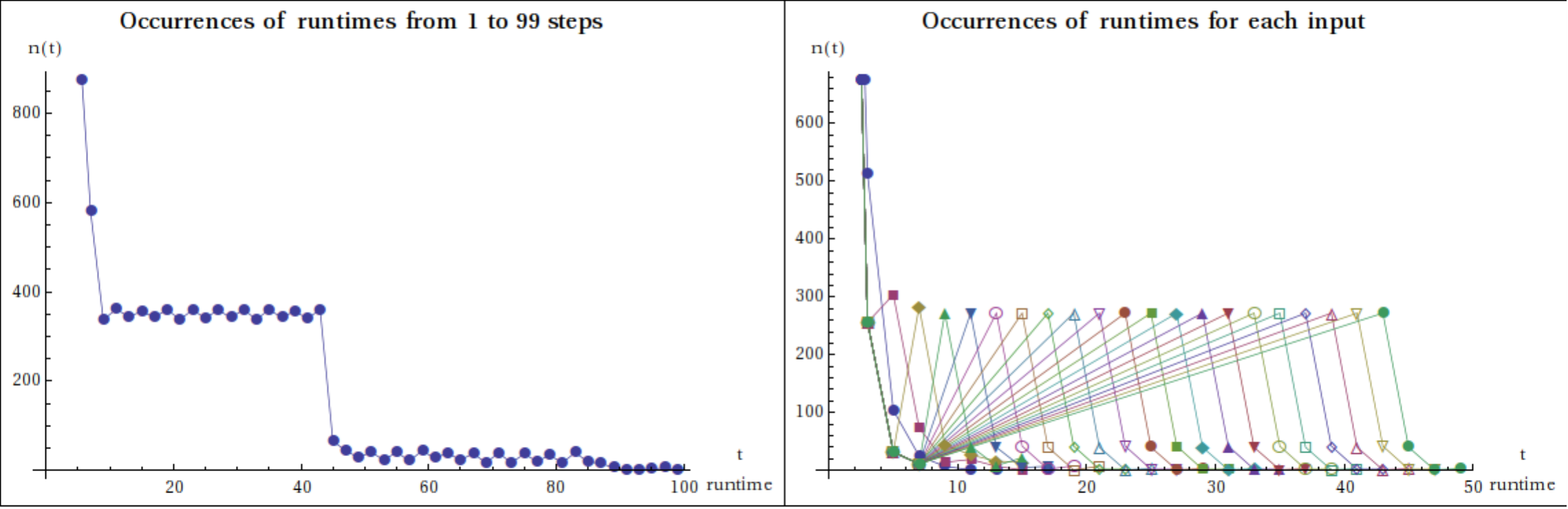}
  \caption{Occurrences of runtimes}
  \label{fig:occrun}
\end{figure}
The phase-transition is apparently caused because there are some
blocks in the runtimes. 
To study the cause of this phase-transition we should observe that the
left diagram on Figure~\ref{fig:occrun} represents the occurrences of
runtimes for arbitrary inputs from 0 to 20. The graph on the right
of Figure~\ref{fig:occrun} is clearer. Now, lines correspond to different inputs
from 0 to 20.
The graph at the left can be obtained from
the right one by adding the occurrences corresponding to points with
the same runtime.  The distribution that we observe here explains the
phase-transition effect. It's very interesting that in all cases there
is a local maximum with around 300 
occurrences
and after this maximum, the evolution is
very similar. In order to explain this, we look at the following
list\footnote{The dots denote a linear progression (or constant which is a special case of linear).} that represents the 10 most frequent runtime sequences in
$(2,2)$. Every runtime sequence is preceded by the number of TMs computing it:
%
\begin{center}
\footnotesize{}
\begin{tabular}{l@{\,\,}l@{\quad\quad}l@{\,\,}l@{\quad\quad}l@{\,\,}l}
2048 &\{1, 1, \ldots\} & 106 & \{-1, 3, 3, \ldots\} & 20 & \{3, 7, 11,
15, \ldots\}\\ 
1265 & \{-1, -1, \ldots\} & 76 &\{3, -1, -1, \ldots\}  & 20 & \{3, 5,
5, \ldots\}\\ 
264 & \{3, 5, 7, 9, \ldots\} & 38 & \{5, 7, 9, 11, \ldots\}\\
112 & \{3, 3, \ldots\} & 32 & \{5, 3, 3, \ldots\} \\

\end{tabular}
\normalsize{}
\end{center}

We observe that there are only 5 sequences computed more than 100
times. They represent 92.65\% of the TMs in $(2,2)$. There is only one
sequence that is not constant nor divergent (recall that $-1$ represents
divergences) with 264 occurrences: $\{3,5,7,9,\ldots\}$. That runtime
sequence corresponds to TMs that give a walk forth and back over the input tape to run of the tape and halt. This is the most trivial linear sequence and explains the intermediate step in the phase-transition effect. There is also
another similar sequence with 38 occurrences
$\{5,7,9,11,\ldots\}$. Moreover, observe that there is a sequence with
20 occurrences where subsequent runtimes differ by 4 steps. This sequence
$\{3,7,11,15,\ldots\}$ contains alternating values of our original one
$\{3,5,7,9,\ldots\}$ and it explains the zigzag observed in the left
part of Figures~\ref{figure:haltingProbDistr22} and~\ref{fig:occrun}.

Altogether, this analysis accounts for the observed phase transition. In a sense, the analysis reduces the phase transition to the strong presence of linear performers in Figure \ref{timecomp1} together with the facts that on the one hand there are few different kinds of linear performers and on the other hand that each group of similar linear TMs is ``spread out over the horizontal axis'' in Figure \ref{figure:haltingProbDistr22} as each input 0,\ldots, 20 is taken into account.
%
%
%

%
%
%

\subsection{Runtimes}
There is a total of 49 different sequences of runtimes in
(2,2). This number is 35 when we only consider total
functions. Most of the runtimes grow linear with the size of the
input. A couple of them grow quadratically and just two grow
exponentially. The longest halting runtime occurs in TM numbers 378
and 1351, that run for 8\,388\,605 steps on the last input, that is on
input 20. Both TMs used only 21 cells\footnote{It is  an interesting question how many times each cell is visited. Is the distribution uniform over the cells? Or centered around the borders?} for their computation and outputted the value 2\,097\,151.

%
%

Rather than exposing lists of outputvalues we shall prefer to graphically
present our data. The sequence of output values is graphically represented as
follows. On the top line we depict the tape output on input zero
(that is, the input consisted of just one black cell). On the second
line immediately below the first one, we depict the tape output on input one (that is, the input
consisted of two black cells), etc. By doing so, we see that the
function computed by TM 378 is just the tape identity.

Let us focus on all the (2,2) TMs that compute that tape
identity. We will depict most of the important information in one
overview diagram. This diagram as shown in
Figure~\ref{figure:OverviewTapeIdentity} contains at the top a
graphical representation of the function computed as described above.

Below the representation of the function, there are six graphs. On
each horizontal axis of these graphs, the input is plotted. The
$\tau_i$ is a diagram that contains plots for all the runtimes of all
the different algorithms computing the function in question. Likewise,
$\sigma_i$ depicts all the space-usages occurring. The ${<} \tau{>}$ and ${<} \sigma{>}$ refer to the
(arithmetical) average of time and space usage. The
subscript $h$ in e.g. ${<}\tau{>}_h$ indicates that the harmonic average is calculated. As
the harmonic average is only defined for non-zero numbers, for
technical reasons we depict the harmonic average of $\sigma_i +2$
rather than for $\sigma_i$.

Let us recall a definition of the harmonic mean. The harmonic mean of $n$ non-zero values $x_1,\ldots, x_n$ is defined as 
\[
{<}x{>}_h := \frac{n}{\frac{1}{x_1}+ \ldots + \frac{1}{x_n}}.
\]
In our case, the harmonic mean of the runtimes can be interpreted as follows. Each TM computes the same function. Thus, the total amount of information in the end computed by each TM per input is the same although runtimes may be different. Hence the runtime of one particular TM on one particular input can be interpreted as time/information. We now consider the following situation:

Let the exhaustive list of TMs computing a particular function $f$ be \{$TM_1$, \dots, $TM_n$ with runtimes $t_1,\\ \ldots, t_n$\}.
If we normalize the amount of information computed by $f$ to 1, we can interpret e.g. $\frac{1}{t_k}$ as the amount of information computed by TM$_k$ in one time step.
If we now let $\mbox{TM}_1$ run for 1 time unit, next $\mbox{TM}_2$ for 1 time unit and finally $\mbox{TM}_n$ for 1 time unit, then the total amount of information of the output computed is $1/t_1 + ... + 1/t_n$. Clearly, 
\[
\stackrel{n \mbox{ \small{times}}}{\overbrace{\frac{1}{{<}\tau{>}_h}+\ldots+\frac{1}{{<}\tau{>}_h}}}=\stackrel{n \mbox{ \small{times}}}{\overbrace{\frac{\frac{1}{t_1}+\ldots+\frac{1}{t_n}}{n}+\ldots\frac{\frac{1}{t_1}+\ldots+\frac{1}{t_n}}{n}}} = \frac{1}{t_1}+\ldots+\frac{1}{t_n}.
\] 
Thus, we can see the harmonic average as the time by which the typical amount of information is gathered on a random TM that computes $f$. Alternatively, the harmonic average ${<}\tau{>}_h$ is such that $\frac{1}{{<}\tau{>}_h}$ is the typical amount of information computed  in one time step on a random TM that computes $f$.

\begin{figure}
  \begin{multicols}{2}
    The image provides the basic information of the TM outputs depicted by a diagram with each row the output of each of the 21 inputs, followed by the plot figures of the average resources taken to compute the function, preceded by the time and space plot for each of the algorithm computing the function. For
    example, this info box tells us that there are 1\,055
    TMs computing the identity function, and that these TMs are distributed over just 12
    different algorithms (i.e. TMs that take different space/time resources).  Notice that at first glance at the
    runtimes $\tau_i$, they seem to follow just an exponential
    sequence while space grows linearly. However, from the other
    diagrams we learn that actually most TMs run in constant time and
    space. Note that all TMs that run out of the tape in the first
    step without changing the cell value (the 25\% of the total space)
    compute this function.

    \begin{center}
      \includegraphics[height=6cm]{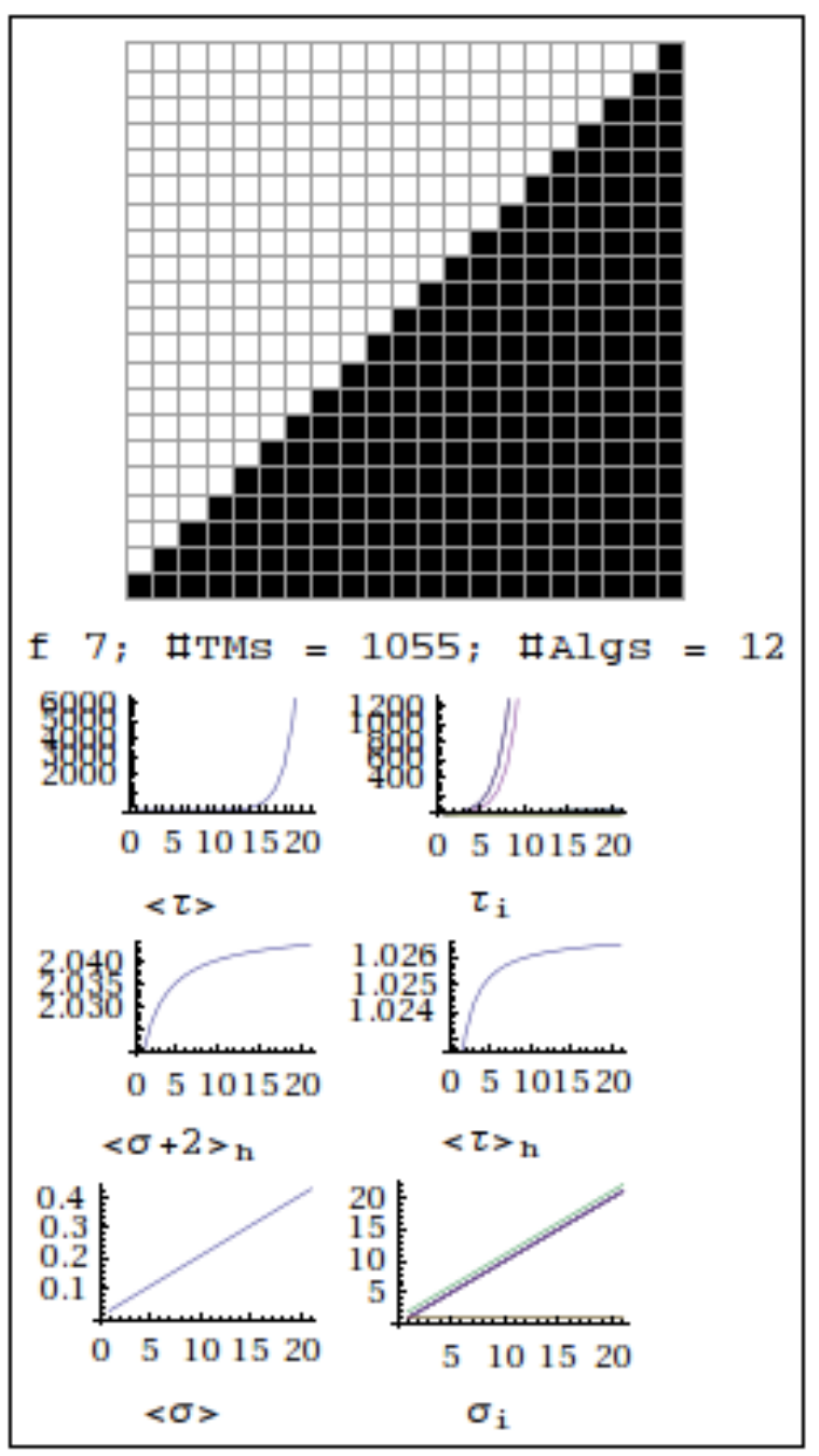}\\
    \end{center}
  \end{multicols}
  \caption{Overview diagram of the tape identity.}
  \label{figure:OverviewTapeIdentity}
\end{figure}

\subsection{Clustering in runtimes and space-usages}
Observe the two graphics in Figure~\ref{fig:SpaceTime22}.
\begin{figure}[htb!]
  \centering
    \includegraphics[width=5.5cm]{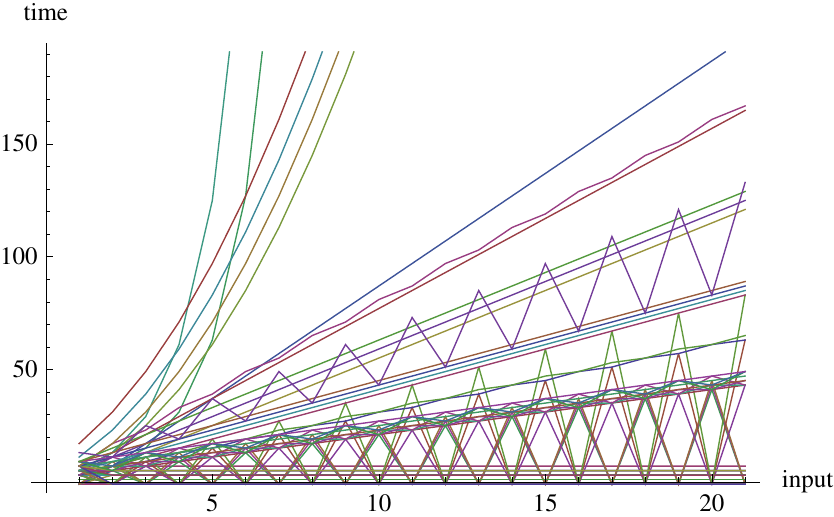}
    \includegraphics[width=5.5cm]{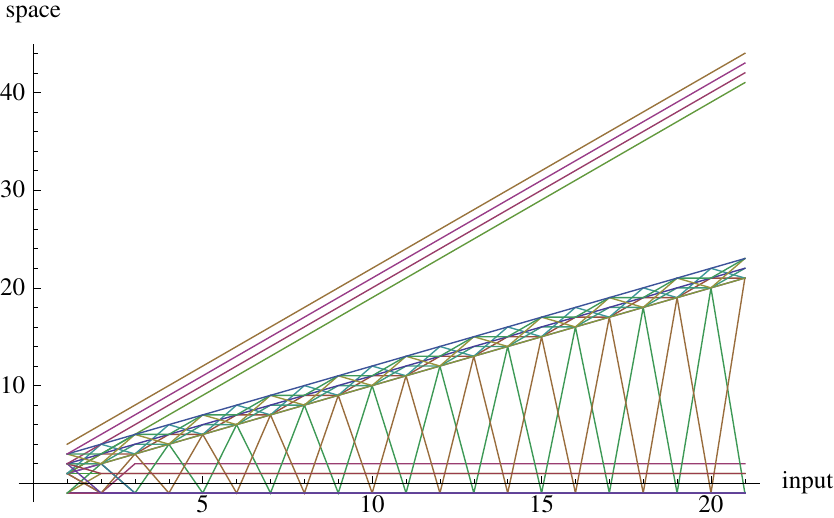}  
  \caption{Runtime and space distribution in (2,2).}
  \label{fig:SpaceTime22}
\end{figure}
The left one shows all the runtime sequences in (2,2) and the
right one the used-space sequences. Divergences are represented by $-1$, so
they explain the values below the horizontal axis. We find some
exponential runtimes and some quadratic ones, but most of them remain linear. All space usages in (2,2) are linear.

An interesting feature of Figure~\ref{fig:SpaceTime22} is the clustering. For example, we see that the space usage comes in three different clusters. 
The clusters are also present in the time graphs. Here the clusters are less prominent as there are more runtimes and the clusters seem to overlap. It is tempting to think of this clustering as rudimentary manifestations of the computational complexity classes.

Another interesting phenomenon is observed in these graphics. It is that of
alternating divergence, detected in those cases where value $-1$
alternates with the other outputs, spaces or runtimes. The phenomena of alternating divergence is also manifest in the study of definable sets.

\subsection{Definable sets}
Like in classical recursion theory, we say that a set $W$ is definable by a (2,2) TM if there is some machine $M$ such that $W=W_M$ where $W_M$ is defined as usual as
\[
W_M := \{x \mid M(x) \downarrow  \}.
\]
In total, there are 8 definable sets in (2,2). Below follows an enumeration of them.
{\small \begin{verbatim}
{{}, {0, 1, 2, 3, 4, 5, 6, 7, 8, 9, 10, 11, 12, 13, 14,
 15, 16, 17, 18, 19, 20}, {0}, {0, 2, 4, 6, 8, 10, 12, 14,
 16, 18, 20}, {1, 2, 3, 4, 5, 6, 7, 8, 9, 10, 11, 12, 13, 
 14, 15, 16, 17, 18, 19, 20}, {2, 3, 4, 5, 6, 7, 8, 9, 10, 
 11, 12, 13, 14, 15, 16, 17, 18, 19, 20}, {1, 3, 5, 7, 9,
 11, 13, 15, 17, 19}, {0, 1}}
\end{verbatim}}
\noindent
It is easy to see that the definable sets are closed under complements.


\subsection{Clustering per function}
We have seen that all runtime sequences in (2,2) come in clusters and likewise for the space usage. It is an interesting observation that this clustering also occurs on the level of single functions. Some examples are reflected in Figure~\ref{figure:clusteringPerFunction}.

\begin{figure}
\begin{center}
\includegraphics[width=10cm]{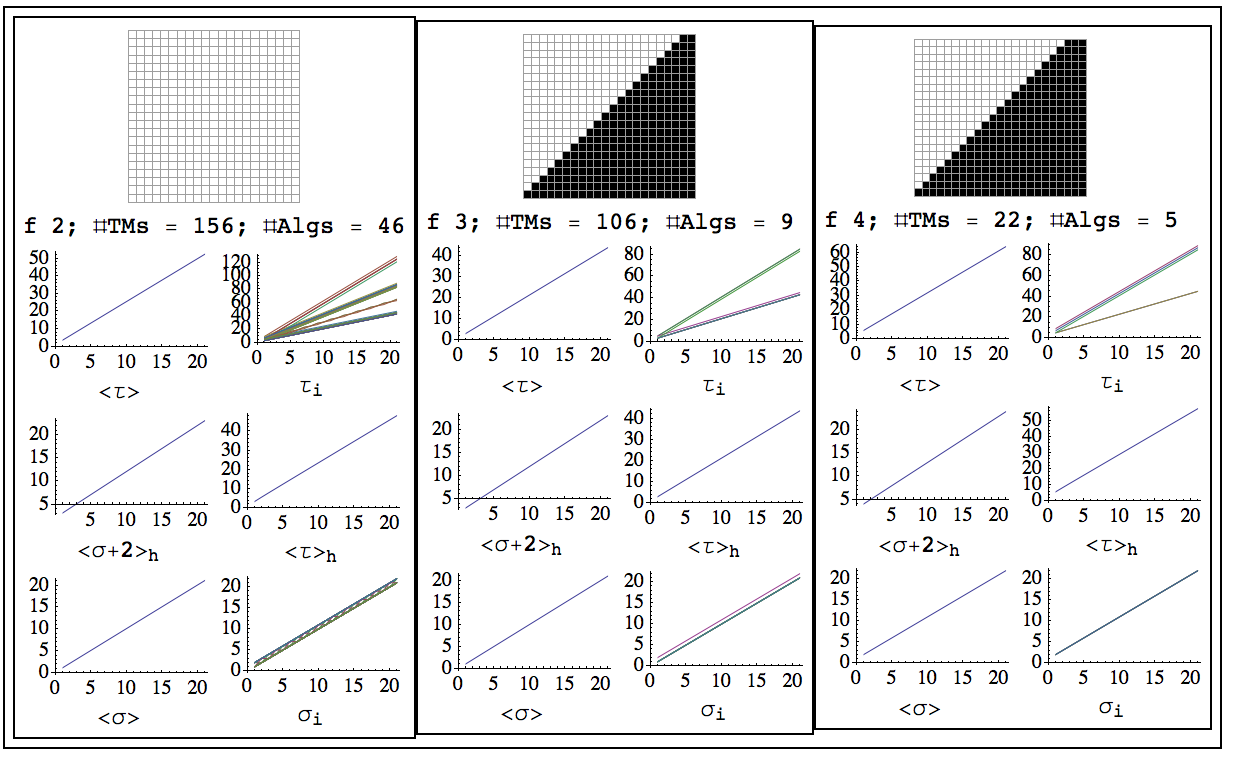}
\caption{Clustering of runtimes and space-usage per function.}\label{figure:clusteringPerFunction}
\end{center}
\end{figure}

\subsection{Computational figures reflecting the number of available resources}
Certain functions clearly reflect the fact that there are only two available states. This is particularly noticeable from the period of alternating converging and non-converging values and in the offset of the growth of the output, and in the alternation period of black and white cells. Some examples are included in Figure~\ref{figure:numberOfStatesManifest}.

\begin{figure}
\begin{center}
\includegraphics[width=10.7cm]{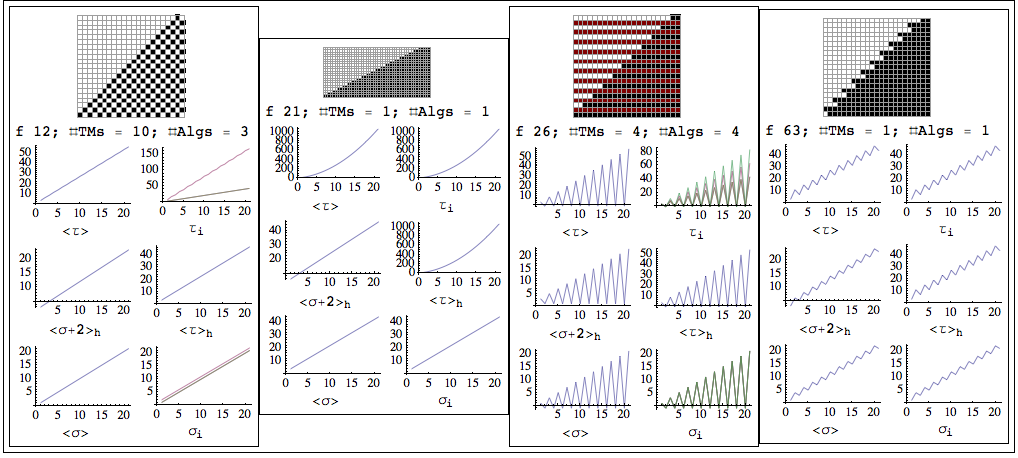}
\caption{Computational figures reflecting the number of available resources.}\label{figure:numberOfStatesManifest} 
\end{center}
\end{figure}

\subsection{Types of computations in (2,2)}

Let us finish this analysis with some comments about the computations
that we can find in (2,2). Most of the TMs perform very simple
computations. Apart from the 50\% that in every space finishes the
computations in just one step (those TMs that move to the right from the
initial state), the general pattern is to make just one round through
the tape and back. It is the case for TM number 2240 with the sequence
of runtimes: 
\begin{verbatim}
{5, 5, 9, 9, 13, 13, 17, 17, 21, 21, ..}
\end{verbatim}
Figure~\ref{fig:tape2240} shows the sequences of tape configurations
for inputs 0 to 5. Each of these five diagrams should be interpreted as follows. The top line represents the tape input and each subsequent line below that represents the tape configuration after one more step in the computation.
\begin{figure}[htb!]
  \centering
  \includegraphics[width=5.8cm]{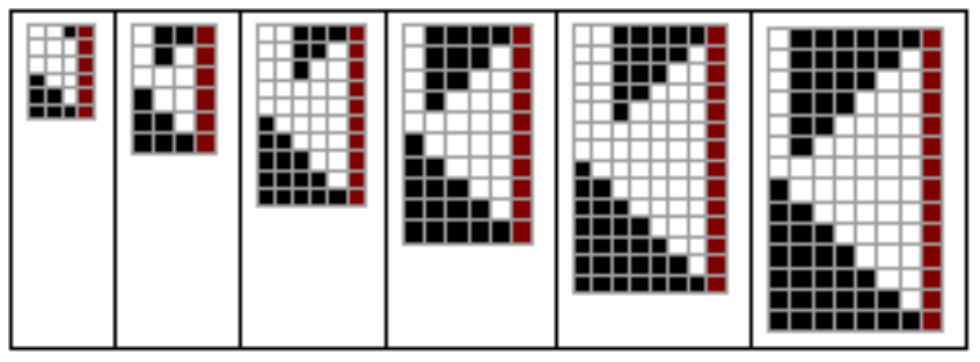}
  \caption{Turing machine tape evolution for Rule 2240.}
  \label{fig:tape2240}
\end{figure}


 The walk around the tape can be more
complicated.  This is the case for TM number 2205 with the runtime
sequence:  
\begin{verbatim}
{3, 7, 17, 27, 37, 47, 57, 67, 77, ...}
\end{verbatim}
which has a greater runtime but it only uses that part of the tape that was given as input, as we can
see in the computations (Figure~\ref{fig:tape2205and1351}, left). 
TM 2205 is interesting in that it shows a clearly localized and propagating pattern that contains the essential computation.

\begin{figure}[htb!]
  \centering
  \includegraphics[width=5cm]{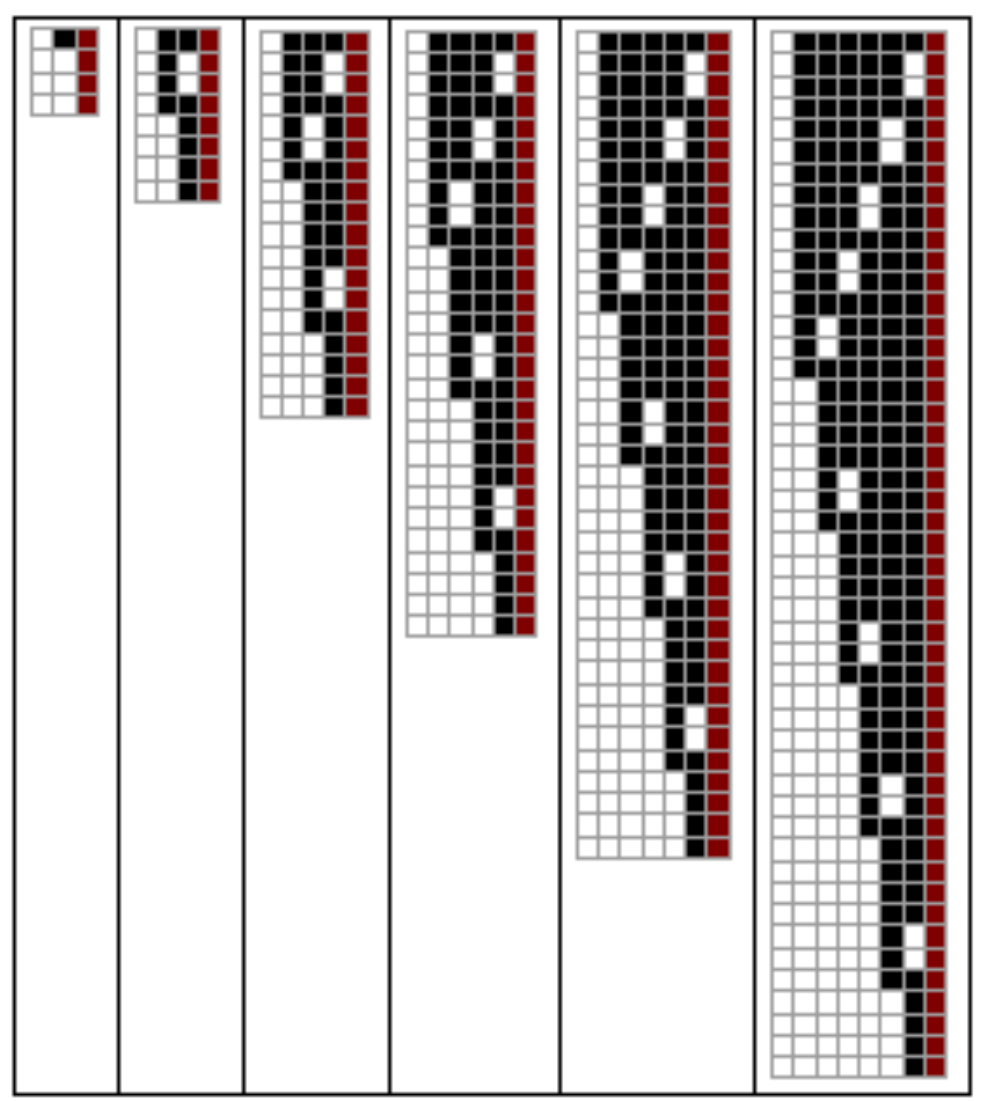}
  \includegraphics[width=3.2cm]{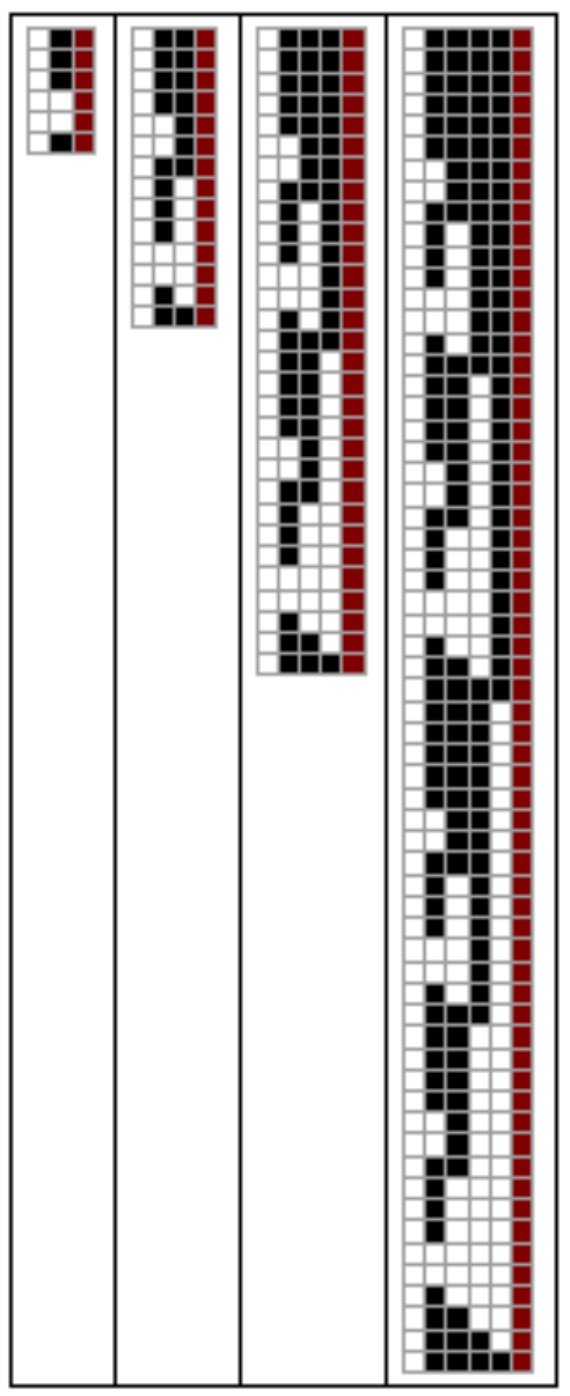}
  \caption{Tape evolution for Rules 2205 (left) and 1351 (right).}
  \label{fig:tape2205and1351}
\end{figure}

The case of TM 1351 is one of the few that escapes from this simple
behavior. As we saw, it has the highest runtimes in
(2,2). Figure~\ref{fig:tape2205and1351} (right) shows its tape
evolution. Note that it is computing the tape identity. Many other TMs in
(2,2) compute this function in linear or constant time. 
In this case of TM 1351 the pattern is generated by a genuine recursive process thus explaining the exponential runtime.

In (2,2) we also witnessed TMs performing iterative computations that gave rise to mainly quadratic runtimes.
An example of this is TM 1447, whose computations for the first seven inputs are represented in Figure~\ref{Figure:IterativeComputation}.

\begin{figure}[htb!]
  \centering
  \includegraphics[width=10.7cm]{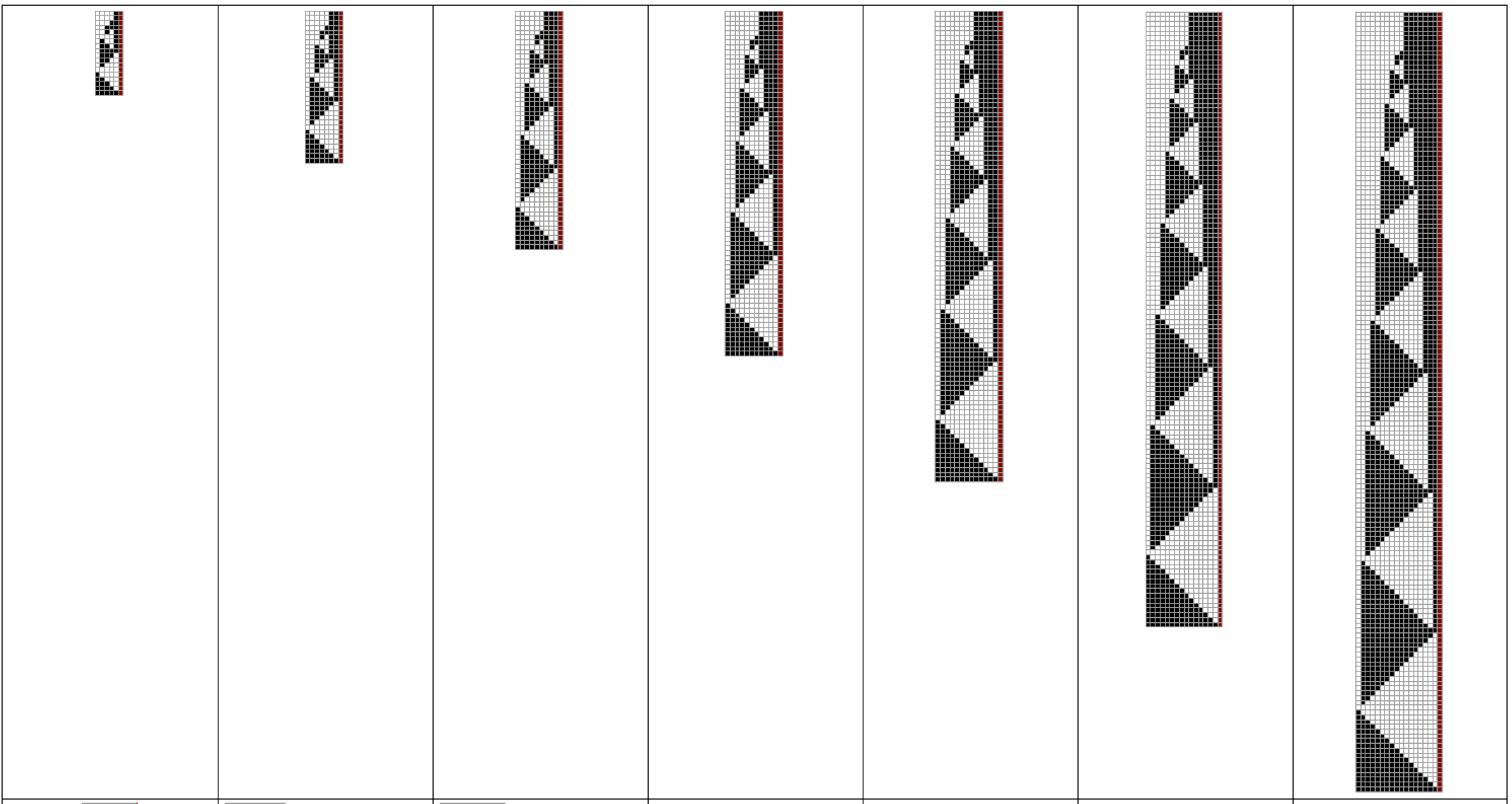}
  \caption{Turing machine tape evolution for Rule 1447.}
  \label{Figure:IterativeComputation}
\end{figure}

Let us briefly summarize the types of computations that we saw in (2,2).

\begin{itemize}
\item 
Constant time behavior like the head (almost) immediately dropping off the tape;

\item
Linear behavior like running to the end of the tape and then back again as Rule 2240;

\item
Iterative behavior like using each black cell to repeat a certain process as in Rule 1447;

\item
Localized computation like in Rule 2205;

\item
Recursive computations like in Rule 1351.

\end{itemize}

As most
of the TMs in (2,2) compute their functions in the easiest possible way (just
one crossing of the tape), no significant speed-up can be expected. Only
slowdown is possible in most cases.

\section{Investigating the space of 3-states, 2-colors Turing machines}\label{Section:32}

In the cleansed data of (3,2) we found 3886  functions and a total of 12824  different algorithms that computed them.

\subsection{Determinant initial segments}

As these machines are more
complex than those of (2,2), more outputs are needed to
characterize a function. From 3 required in (2,2) we need now
8, see Figure~\ref{fig:grphn}.
\begin{figure}[htb!]
  \centering
  \includegraphics[width=6cm]{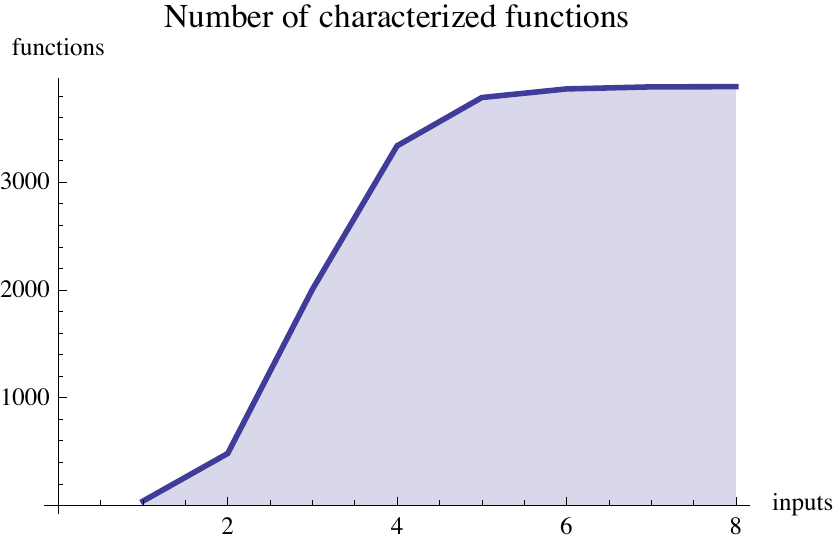}
  \caption{Number of outputs required to characterize a function in (3,2).}
  \label{fig:grphn}
\end{figure}

\subsection{Halting probability}

Figure~\ref{fig:runProb32} shows the runtime probability distributions
in (3,2). The same behavior that we commented for (2,2)
is also observed. 
\begin{figure}[htb!]
  \centering
  \includegraphics[width=10.1cm]{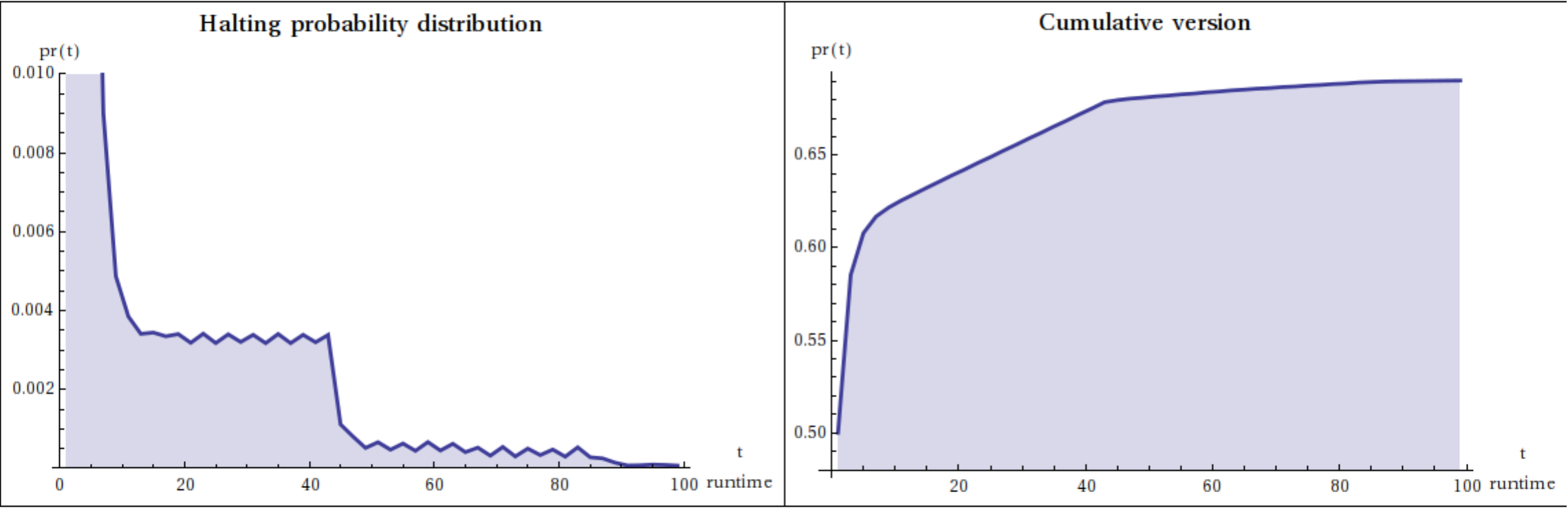}
  \caption{Runtime proprobability distributions in (3,2).}
  \label{fig:runProb32}
\end{figure}

Note that the ``phase transitions" in (3,2) are even more pronounced than in (2,2).
We can see these phase transitions as rudimentary manifestations of computational complexity classes.
Similar reasoning as in Subsection \ref{Section:PhaseTransitions} can be applied for (3,2) to account for the phase transitions as we can see in Figure~\ref{fig:occrun32}.

\begin{figure}[htb!]
  \centering
  \includegraphics[width=10.1cm]{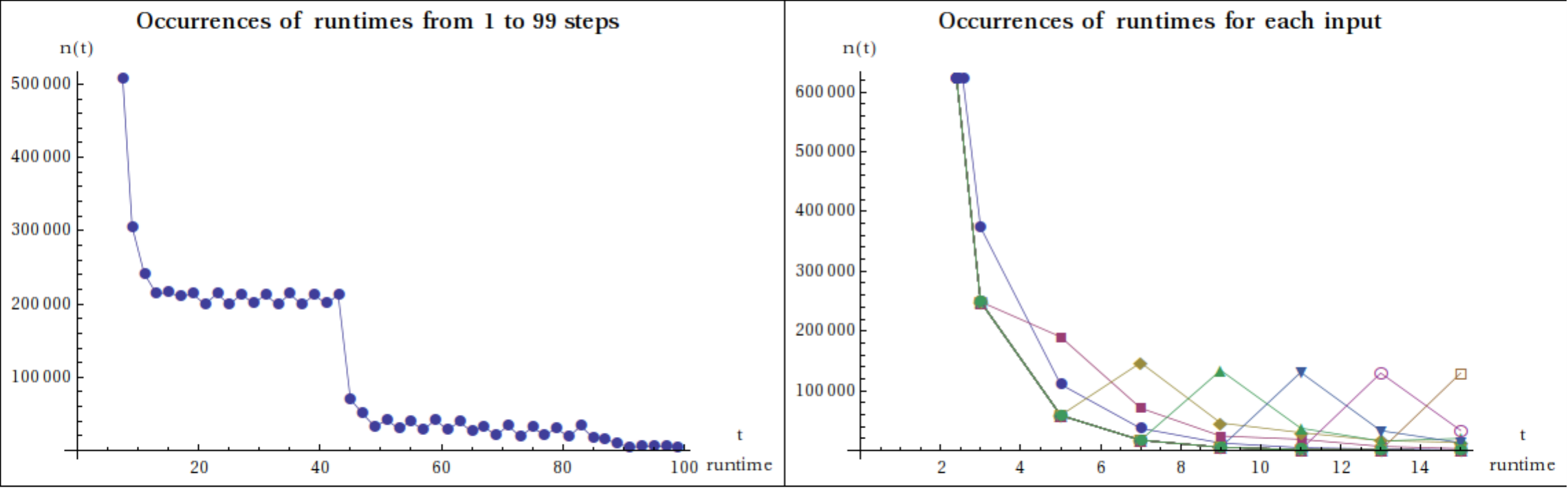}
\caption{Occurrences of runtimes}
  \label{fig:occrun32}
\end{figure}

\subsection{Runtimes and space-usages}

In (3,2) the number of different runtimes and space usage sequences is
the same: 3676. Plotting them all as we did for
(2,2) would not be too informative in this case. So, Figure~\ref{fig:SpaceTime32sample} shows samples of
50 sequences of space and runtime sequences. Divergent values are
omitted as to avoid big sweeps in the graphs caused by the alternating divergers. As in (2,2) we observe the same phenomenon of clustering. 
\begin{figure}[htb!]
  \centering
  \includegraphics[width=5cm]{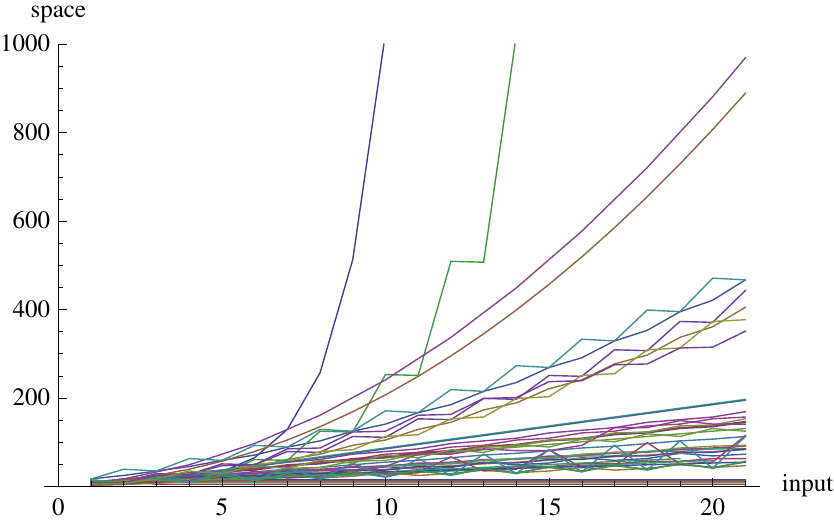}
  \includegraphics[width=5cm]{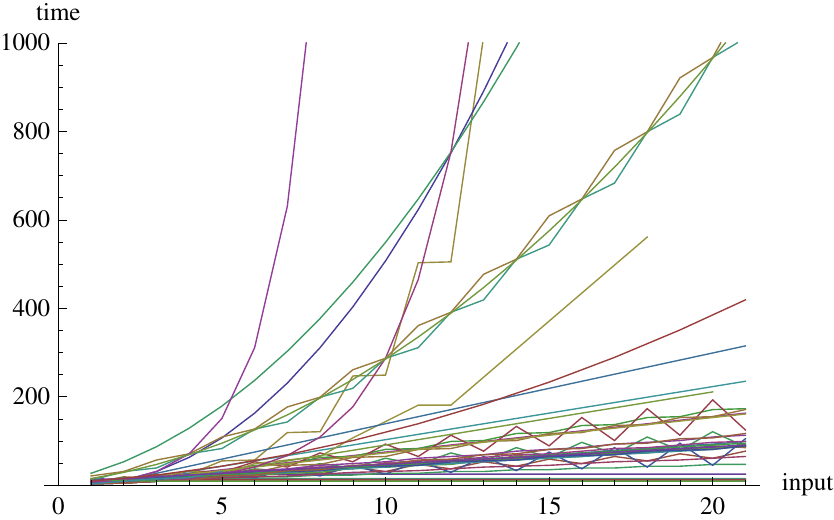}  
  \caption{Sampling of 50 space (left) and runtime (right) sequences
    in (3,2).} 
  \label{fig:SpaceTime32sample}
\end{figure}

\subsection{Definable sets}

Now we have found 100 definable sets. Recall that in (2,2)
definable sets were closed under taking complements. This does not happen
in (3,2). There are 46 definable sets, like
{\small \begin{verbatim}
{{}, {0}, {1}, {2}, {0, 1}, {0, 2}, {1, 2}, 
 {0, 1, 2}, ...}
\end{verbatim}}
\noindent
that coexist with their complements, but another 54, like
{\small \begin{verbatim}
{{0, 3}, {1, 3}, {1, 4}, {0, 1, 4}, {0, 2, 3}, 
 {0, 2, 4}, ...}
\end{verbatim}}
\noindent
are definable sets but their complements are not. 
We note that, although there are more definable sets in (3,2) in an absolute sense, the number of definable sets in (3,2) relative to the total amount of functions in (3,2) is about four times smaller than in (2,2).

\subsection{Clustering per function}

In (3,2) the same phenomenon of the clustering of runtime and
space usage within a single function also happens. 
\begin{figure}[htb!]
  \centering
  \includegraphics[width=9cm]{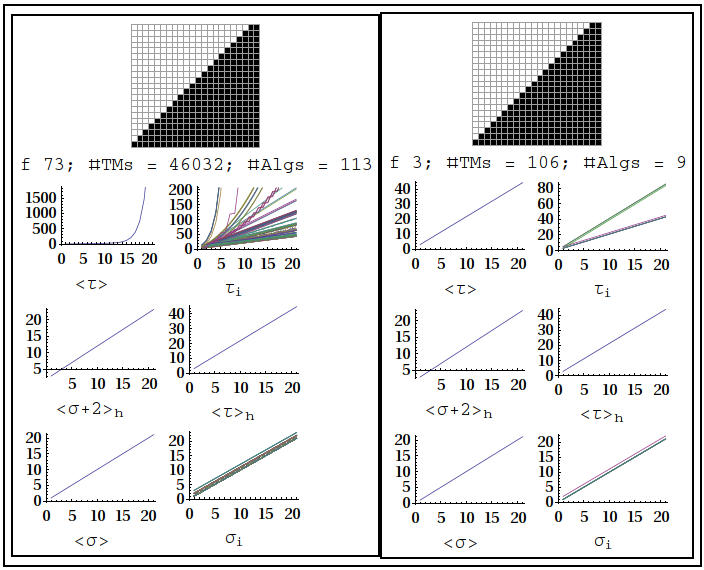}
  \caption{Clustering per function in (3,2).}
  \label{fig:clustFunc32}
\end{figure}
Moreover, as Figure~\ref{fig:clustFunc32} shows, exponential runtime sequences may
occur in a (3,2) function (left) while only linear behavior is present among the (2,2) computations of the
function (right). 

\subsection{Exponential behavior in (3,2) computations}

Recall that in (2,2) most convergent TMs complete their
computations in linear time. Now (3,2) presents more interesting
exponential behavior, not only in runtime but also in used space. 

The max runtime in (3,2) is 894\,481\,409 steps found in the TMs number
599063 and 666364 (a pair of twin rules\footnote{We call two rules in (3,2) \emph{twin rules} whenever they are exactly the same after switching the role of State 2 and State 3. })  at input
20. The values of this function are double exponential. All of them
are a power of 2 minus 2. 

%
%
%
%
%
Figure~\ref{fig:tape599} shows the tape evolution with
inputs 0 and 1. The pattern observed on the right repeats itself.
\begin{figure}[htb!]
  \centering
  \includegraphics[width=2.5cm]{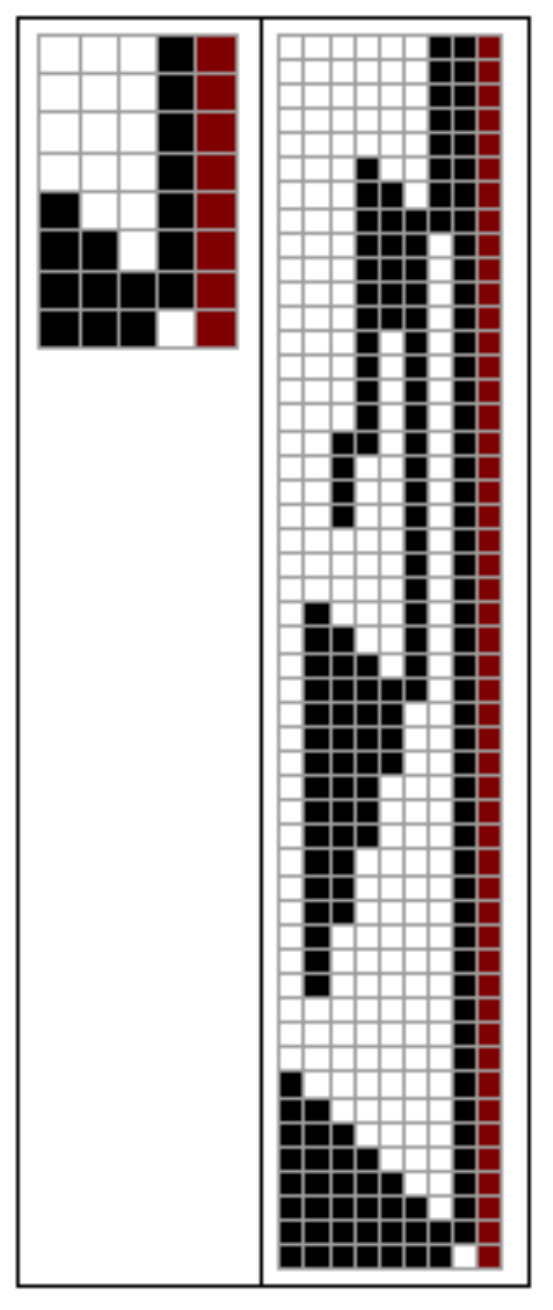}
  \caption{Tape evolution for Rule 599063.}
  \label{fig:tape599}
\end{figure}

\section{The space (4,2)}\label{Section:42}
An exhaustive search of this space fell out of the scope of the current project. For the sake of our investigations we were merely interested in finding functions in (4,2) that we were interested in. Thus, we sampled and looked only for interesting functions that we selected from (2,2) and (3,2). 
In searching the 4,2 space, we proceeded as follows. We selected 284 functions in (3,2), 18 of them also in (2,2), that we hoped to find in (4,2) using a sample of about 56x$10^6$ random TMs. 

Our search process consisted of generating random TMs and run them for 1000 steps, with inputs from 0 to 21. The output (with runtime and space usage) was saved only for those TMs with a converging part that matches some of the 284 selected functions. 

We saved 32235683 TMs. From these, 28032552 were very simple TMs that halt in just one step for every input, so we removed them. We worked with 4203131 non-trivial TMs. 

After cleansing there were 1549 functions computed by 49674 algorithms. From these functions, 22 are in (2,2) and 429 in (3,2). TMs computing all the 284 functions of the sampling were found.

Throughout the remainder of the paper it is good to constantly have in mind that the sampling in the (4,2) space is not at all representative.

\section{Comparison between the TM spaces}\label{Section:Comparing}

The most prominent conclusion from this section is that when computing a particular function, slow-down of a computation is more likely than speed-up if the TMs have access to more resources to perform their computations. Actually no essential speed-up was witnessed at all. We shall compare the runtimes both numerically and asymptotically.

\subsection{Runtimes comparison}

In this section we compare the types of runtime progressions we encountered in our experiment. We use the big $\mathcal{O}$ notation to classify the different types of runtimes. Again, it is good to bear in mind that our findings are based on just 21 different inputs. However, the estimates of the asymptotic behavior is based on the functions as found in the cleansing process and sanity checks on more inputs confirmed the correctness (plausibility) of those functions.

Below, a table is presented that compares the runtime behavior between functions that were present in (2,2), (3,2) and (4,2). The first column refers to a canonical index for this list of functions that live in all of (2,2), (3,2) and (4,2). The column under the heading (2,2) displays the distribution of time complexity classes for the different algorithms  in (2,2) computing the particular function in that row and likewise for the columns (3,2) and (4,2). Each time complexity class is followed by the number of occurrences among the algorithms in that TM space. The complexity classes are sorted in increasing order. Note that we only display a selection of the functions, but our selection is representative for the whole (2,2) space.

 {\tiny
\flushleft{
\[
\begin{array}{|c|c|c|c|}
\hline
\text{ ${\#}$} & {\text{(2,2)}} & \text{(3,2)} & \text{(4,2)} \\
\hline
1 & 
\begin{array}{lll}
\text{O}[1]:46 & \text{O}[n]:46 & \text{}
\end{array}
& 
\begin{array}{lll}
\text{O}[1]:1109 & \text{O}[n]:1429 \\
\text{O}\left[n^2\right]:7 & \text{O}\left[n^3\right]:1 & \text{} 
\end{array}
& 
\begin{array}{lll}
\text{O}[1]:19298 & \text{O}[n]:28269 \\
 \text{O}\left[n^2\right]:77 & \text{O}\left[n^3\right]:6 & \text{}
\end{array}
\\
\hline
2 & 
\begin{array}{lll}
\text{O}[1]:5 & \text{O}[n]:5 & \text{}
\end{array}
& 
\begin{array}{lll}
\text{O}[1]:73 & \text{O}[n]:64 \\
\text{O}\left[n^2\right]:7 & \text{O}[\text{Exp}]:4 & \text{}
\end{array}
& 
\begin{array}{lll}
\text{O}[1]:619 & \text{O}[n]:566 \\
\text{O}\left[n^2\right]:53 & \text{O}\left[n^3\right]:16 \\
\text{O}[\text{Exp}]:26 & \text{}
\end{array}
\\
\hline
3 & 
\begin{array}{lll}
\text{O}[1]:2 & \text{O}[n]:2 & \text{}
\end{array}
& 
\begin{array}{lll}
\text{O}[1]:129 & \text{O}[n]:139 \\
 \text{O}\left[n^2\right]:2 & \text{}
\end{array}
& 
\begin{array}{lll}
\text{O}[1]:2483 & \text{O}[n]:3122 \\
 \text{O}\left[n^2\right]:68 & \text{O}\left[n^3\right]:1
\end{array}
\\
\hline
4 & 
\begin{array}{lll}
\text{O}[1]:16 & \text{O}[n]:5 \\
 \text{O}[\text{Exp}]:3 & \text{}
\end{array}
& 
\begin{array}{lll}
\text{O}[1]:124 & \text{O}[n]:34 \\
 \text{O}\left[n^2\right]:9 & \text{O}\left[n^3\right]:15 \\
  \text{O}\left[n^4\right]:5 & \text{O}[\text{Exp}]:15 & \text{}
\end{array}
& 
\begin{array}{lll}
\text{O}[1]:1211 & \text{O}[n]:434 \\
\text{O}\left[n^2\right]:101 & \text{O}\left[n^3\right]:181 \\
\text{O}\left[n^4\right]:59 & \text{O}[\text{Exp}]:156
\end{array}
\\
\hline
5 & 
\begin{array}{lll}
\text{O}[1]:2 & \text{O}[n]:2 & \text{}
\end{array}
& 
\begin{array}{lll}
\text{O}[1]:34 & \text{O}[n]:34 & \text{}
\end{array}
& 
\begin{array}{lll}
\text{O}[1]:289 & \text{O}[n]:285 \\
\text{O}\left[n^2\right]:8 & \text{}
\end{array}
\\
\hline
6 & 
\begin{array}{lll}
\text{O}[1]:3 & \text{O}[n]:3 & \text{}
\end{array}
& 
\begin{array}{lll}
\text{O}[1]:68 & \text{O}[n]:74 & \text{}
\end{array}
& 
\begin{array}{lll}
\text{O}[1]:576 & \text{O}[n]:668 \\
\text{O}\left[n^2\right]:9 & \text{O}\left[n^3\right]:3
\end{array}
\\
\hline
7 & 
\begin{array}{lll}
\text{O}[1]:10 & \text{} & \text{}
\end{array}
& 
\begin{array}{lll}
\text{O}[1]:54 & \text{O}[n]:8 & \text{}
\end{array}
& 
\begin{array}{lll}
\text{O}[1]:368 & \text{O}[n]:94 \\
\text{O}\left[n^3\right]:4 & \text{O}[\text{Exp}]:6
\end{array}
\\
\hline
8 & 
\begin{array}{lll}
\text{O}[n]:1 & \text{O}\left[n^2\right]:1 & \text{}
\end{array}
& 
\begin{array}{lll}
\text{O}[n]:13 & \text{O}\left[n^2\right]:13 & \text{}
\end{array}
& 
\begin{array}{lll}
\text{O}[n]:112 & \text{O}\left[n^2\right]:107 \\
\text{O}\left[n^3\right]:4 & \text{O}[\text{Exp}]:1
\end{array}
\\
\hline
9 & 
\begin{array}{lll}
\text{O}[1]:2 & \text{O}[n]:2 & \text{}
\end{array}
& 
\begin{array}{lll}
\text{O}[1]:58 & \text{O}[n]:54 \\
 \text{O}\left[n^2\right]:4 & \text{}
\end{array}
& 
\begin{array}{lll}
\text{O}[1]:503 & \text{O}[n]:528 \\
\text{O}\left[n^2\right]:23 & \text{O}[\text{Exp}]:4
\end{array}
\\
\hline
10 & 
\begin{array}{lll}
\text{O}[n]:1 & \text{O}\left[n^2\right]:1 & \text{}
\end{array}
& 
\begin{array}{lll}
\text{O}[n]:11 & \text{O}\left[n^2\right]:11 & \text{}
\end{array}
& 
\begin{array}{lll}
\text{O}[n]:114 & \text{O}\left[n^2\right]:110 \\
\text{O}\left[n^3\right]:1 & \text{O}[\text{Exp}]:3
\end{array}
\\
\hline
11 & 
\begin{array}{lll}
\text{O}[n]:1 & \text{O}\left[n^2\right]:1 & \text{}
\end{array}
& 
\begin{array}{lll}
\text{O}[n]:11 & \text{O}\left[n^2\right]:11 & \text{}
\end{array}
& 
\begin{array}{lll}
\text{O}[n]:91 & \text{O}\left[n^2\right]:88 \\
\text{O}\left[n^3\right]:1 & \text{O}[\text{Exp}]:2 
\end{array}
\\
\hline
12 & 
\begin{array}{lll}
\text{O}[n]:1 & \text{O}\left[n^2\right]:1 & \text{}
\end{array}
& 
\begin{array}{lll}
\text{O}[n]:12 & \text{O}\left[n^2\right]:12 & \text{}
\end{array}
& 
\begin{array}{lll}
\text{O}[n]:120 & \text{O}\left[n^2\right]:112 \\
\text{O}\left[n^3\right]:3 & \text{O}[\text{Exp}]:5
\end{array}
\\
\hline
13 & 
\begin{array}{lll}
\text{O}[1]:5 & \text{O}[n]:5 & \text{}
\end{array}
& 
\begin{array}{lll}
\text{O}[1]:39 & \text{O}[n]:43 & \text{}
\end{array}
& 
\begin{array}{lll}
\text{O}[1]:431 & \text{O}[n]:546 \\
 \text{O}\left[n^2\right]:1  & \text{}
\end{array}
\\
\hline
14 & 
\begin{array}{lll}
\text{O}[1]:4 & \text{O}[n]:4 & \text{}
\end{array}
& 
\begin{array}{lll}
\text{O}[1]:14 & \text{O}[n]:14 & \text{}
\end{array}
& 
\begin{array}{lll}
\text{O}[1]:119 & \text{O}[n]:121 \\
\text{O}\left[n^2\right]:5 & \text{O}\left[n^3\right]:1
\end{array}
\\
\hline
15 & 
\begin{array}{lll}
\text{O}[1]:2 & \text{} & \text{}
\end{array}
& 
\begin{array}{lll}
\text{O}[1]:11 & \text{O}[n]:1 & \text{}
\end{array}
& 
\begin{array}{lll}
\text{O}[1]:69 & \text{O}[n]:15 \\
\text{O}\left[n^2\right]:1 & \text{O}[\text{Exp}]:3
\end{array}
\\
\hline
16 & 
\begin{array}{lll}
\text{O}[1]:18 & \text{} & \text{}
\end{array}
& 
\begin{array}{lll}
\text{O}[1]:27 & \text{O}[n]:7 & \\
\text{O}\left[n^3\right]:1 & \text{O}[\text{Exp}]:3
\end{array}
& 
\begin{array}{lll}
\text{O}[1]:233 & \text{O}[n]:63 \\
\text{O}\left[n^2\right]:15 & \text{O}\left[n^3\right]:24 \\
\text{O}\left[n^4\right]:4 & \text{O}[\text{Exp}]:29
\end{array}
\\
\hline
17 & 
\begin{array}{lll}
\text{O}[1]:2 & \text{O}[n]:2 & \text{}
\end{array}
& 
\begin{array}{lll}
\text{O}[1]:33 & \text{O}[n]:33 & \text{}
\end{array}
& 
\begin{array}{lll}
\text{O}[1]:298 & \text{O}[n]:294 \\
\text{O}\left[n^2\right]:2 & \text{O}\left[n^3\right]:2
\end{array}
\\
\hline
18 & 
\begin{array}{lll}
\text{O}[1]:1 & \text{O}[n]:1 & \text{}
\end{array}
& 
\begin{array}{lll}
\text{O}[1]:9 & \text{O}[n]:9 & \text{}
\end{array}
& 
\begin{array}{lll}
\text{O}[1]:94 & \text{O}[n]:94 & \text{}
\end{array}
\\
\hline
19 & 
\begin{array}{lll}
\text{O}[1]:1 & \text{O}[n]:1 & \text{}
\end{array}
& 
\begin{array}{lll}
\text{O}[1]:78 & \text{O}[n]:87 \\
 \text{O}\left[n^2\right]:1 & \text{}
\end{array}
& 
\begin{array}{lll}
\text{O}[1]:1075 & \text{O}[n]:1591 \\
\text{O}\left[n^2\right]:28 & \text{}
\end{array}
\\
\hline
20 & 
\begin{array}{lll}
\text{O}[1]:1 & \text{O}[n]:1 & \text{}
\end{array}
& 
\begin{array}{lll}
\text{O}[1]:15 & \text{O}[n]:15 & \text{}
\end{array}
& 
\begin{array}{lll}
\text{O}[1]:76 & \text{O}[n]:75 \\
\text{O}\left[n^2\right]:1 & \text{}
\end{array}
\\
\hline
21 & 
\begin{array}{lll}
\text{O}[1]:1 & \text{O}[n]:1 & \text{}
\end{array}
& 
\begin{array}{lll}
\text{O}[1]:21 & \text{O}[n]:21 & \text{}
\end{array}
& 
\begin{array}{lll}
\text{O}[1]:171 & \text{O}[n]:173 & \text{}
\end{array}
\\
\hline
22 & 
\begin{array}{lll}
\text{O}[1]:1 & \text{O}[n]:1 & \text{}
\end{array}
& 
\begin{array}{lll}
\text{O}[1]:14 & \text{O}[n]:14 & \text{}
\end{array}
& 
\begin{array}{lll}
\text{O}[1]:203 & \text{O}[n]:203 \\
\text{O}\left[n^2\right]:2 & \text{O}[\text{Exp}]:4
\end{array}\\
\hline
\end{array}
\]}

}

%

No essentially (different asymptotic behavior) faster runtime was found in (3,2) compared to (2,2). Thus, no speed-up was found other than by a linear factor as reported in Subsection  (\ref{linear}). That is, no algorithm in (3,2) computing a function in (2,2) was essentially faster than the fastest algorithm computing the same function in (2,2). Amusing findings were Turing machines both in (2,2) and (3,2) computing the tape identify function in as much as exponential time. They are an example of machines spending all resources to compute a simple function. Another example is the constant function f(n) = 0 computed in $O(n^2)$, $O(n^3)$, $O(n^4)$ and even $O(Exp)$.

\begin{figure}[htbp!]
 \centering
{\tiny
$
\begin{array}{|c|c|c|}
\hline
\text{$\#$} & \text{(3,2)} & \text{(4,2)} \\
\hline
1 &
\begin{array}{lll}
\text{O}[n],1265 & \text{O}\left[n^2\right],7 & \text{O}\left[n^3\right],1
\end{array}
&
\begin{array}{lll}
\text{O}[n],23739 & \text{O}\left[n^2\right],80 & \text{O}\left[n^3\right],6
\end{array}
\\
\hline
2 &
\begin{array}{lll}
\text{O}[n],82 & \text{O}\left[n^2\right],1 & \text{}
\end{array}
&
\begin{array}{lll}
\text{O}[n],1319 & \text{O}\left[n^2\right],28 & \text{}
\end{array}
\\
\hline
4 &
\begin{array}{lll}
\text{O}[n],133 & \text{O}\left[n^2\right],2 & \text{}
\end{array}
&
\begin{array}{lll}
\text{O}[n],2764 & \text{O}\left[n^2\right],72 & \text{O}\left[n^3\right],1
\end{array}
\\
\hline
6 &
\begin{array}{lll}
\text{O}[1],23 & \text{O}[n],34 & \text{O}\left[n^2\right],9 \\
\text{O}\left[n^3\right],15 & \text{O}\left[n^4\right],5 & \text{O}[\text{Exp}],15
\end{array}
&
\begin{array}{lll}
\text{O}[1],197 & \text{O}[n],377 & \text{O}\left[n^2\right],101 \\
\text{O}\left[n^3\right],181 & \text{O}\left[n^4\right],59 & \text{O}[\text{Exp}],156
\end{array}
\\
\hline
10 &
\begin{array}{lll}
\text{O}[n],54 & \text{O}\left[n^2\right],4 & \text{}
\end{array}
&
\begin{array}{lll}
\text{O}[n],502 & \text{O}\left[n^2\right],23 & \text{O}[\text{Exp}],4
\end{array}
\\
\hline
12 &
\begin{array}{lll}
\text{O}\left[n^2\right],11 & \text{} & \text{}
\end{array}
&
\begin{array}{lll}
\text{O}\left[n^2\right],110 & \text{O}\left[n^3\right],1 & \text{O}[\text{Exp}],3
\end{array}
\\
\hline
13 &
\begin{array}{lll}
\text{O}[n],63 & \text{O}\left[n^2\right],7 & \text{O}[\text{Exp}],4
\end{array}
&
\begin{array}{lll}
\text{O}[n],544 & \text{O}\left[n^2\right],54 & \text{O}\left[n^3\right],16 \\
\text{O}[\text{Exp}],26 & \text{} & \text{}
\end{array}
\\
\hline
32 &
\begin{array}{lll}
\text{O}\left[n^2\right],12 & \text{} & \text{}
\end{array}
&
\begin{array}{lll}
\text{O}\left[n^2\right],112 & \text{O}\left[n^3\right],3 & \text{O}[\text{Exp}],5
\end{array}
\\
\hline
95 &
\begin{array}{lll}
\text{O}[1],8 & \text{O}[n],7 & \text{O}\left[n^3\right],1 \\
\text{O}[\text{Exp}],3 & \text{} & \text{}
\end{array}
&
\begin{array}{lll}
\text{O}[1],49 & \text{O}[n],63 & \text{O}\left[n^2\right],15 \\
\text{O}\left[n^3\right],24 & \text{O}\left[n^4\right],4 & \text{O}[\text{Exp}],29
\end{array}
\\
\hline
100 &
\begin{array}{lll}
\text{O}[n],9 & \text{O}\left[n^2\right],1 & \text{}
\end{array}
&
\begin{array}{lll}
\text{O}[n],90 & \text{O}\left[n^2\right],4 & \text{}
\end{array}
\\
\hline
112 &
\begin{array}{lll}
\text{O}[n],5 & \text{O}\left[n^2\right],1 & \text{O}\left[n^3\right],3
\end{array}
&
\begin{array}{lll}
\text{O}[n],41 & \text{O}\left[n^2\right],10 & \text{O}\left[n^3\right],12
\end{array}
\\
\hline
135 &
\begin{array}{lll}
\text{O}\left[n^2\right],13 & \text{} & \text{}
\end{array}
&
\begin{array}{lll}
\text{O}\left[n^2\right],107 & \text{O}\left[n^3\right],4 & \text{O}[\text{Exp}],1
\end{array}
\\
\hline
138 &
\begin{array}{lll}
\text{O}[n],5 & \text{O}\left[n^2\right],3 & \text{O}\left[n^3\right],1
\end{array}
&
\begin{array}{lll}
\text{O}[n],34 & \text{O}\left[n^2\right],13 & \text{O}\left[n^3\right],10
\end{array}
\\
\hline
292 &
\begin{array}{lll}
\text{O}[n],7 & \text{O}\left[n^2\right],1 & \text{}
\end{array}
&
\begin{array}{lll}
\text{O}[n],162 & \text{O}\left[n^2\right],5 & \text{}
\end{array}
\\
\hline
350 &
\begin{array}{lll}
\text{O}[n],1 & \text{O}\left[n^2\right],1 & \text{}
\end{array}
&
\begin{array}{lll}
\text{O}[n],20 & \text{O}\left[n^2\right],1 & \text{O}\left[n^3\right],2
\end{array}
\\
\hline
421 &
\begin{array}{lll}
\text{O}\left[n^2\right],11 & \text{} & \text{}
\end{array}
&
\begin{array}{lll}
\text{O}\left[n^2\right],88 & \text{O}\left[n^3\right],1 & \text{O}[\text{Exp}],2
\end{array}
\\
\hline
422 &
\begin{array}{lll}
\text{O}\left[n^2\right],2 & \text{} & \text{}
\end{array}
&
\begin{array}{lll}
\text{O}\left[n^2\right],17 & \text{} & \text{}
\end{array}\\
\hline
\end{array}
$
}
\caption{Comparison of the distributions of time classes of algorithms computing a particular function for a sample of 17 functions computed both in (3,2) and (4,2). The function number is an index from the list containing all 429 functions considered by us, that were computed in both TM spaces.}
 \label{Figure:TableComplClasses32Versus42}
\end{figure}


In (2,2) however, there are very few non-linear time algorithms and functions\footnote{We call a function $O(f)$ time, when its asymptotically fastest algorithm is $O(f)$ time.}. However as we see from the similar table for (3,2) versus (4,2) in Figure~\ref{Figure:TableComplClasses32Versus42}, also between these spaces there is no essential speed-up witnessed. Again only speed-up by a linear factor can occur.

\subsection{Distributions over the complexity classes}

Figure~\ref{timecomp1} shows the distribution of the the TMs over the  different asymptotic complexity classes. On the level of this distribution we see that the slow-down is manifested in a shift of the distribution to the right of the spectrum.

\begin{figure}
  \centering{}
  \includegraphics[width=10cm]{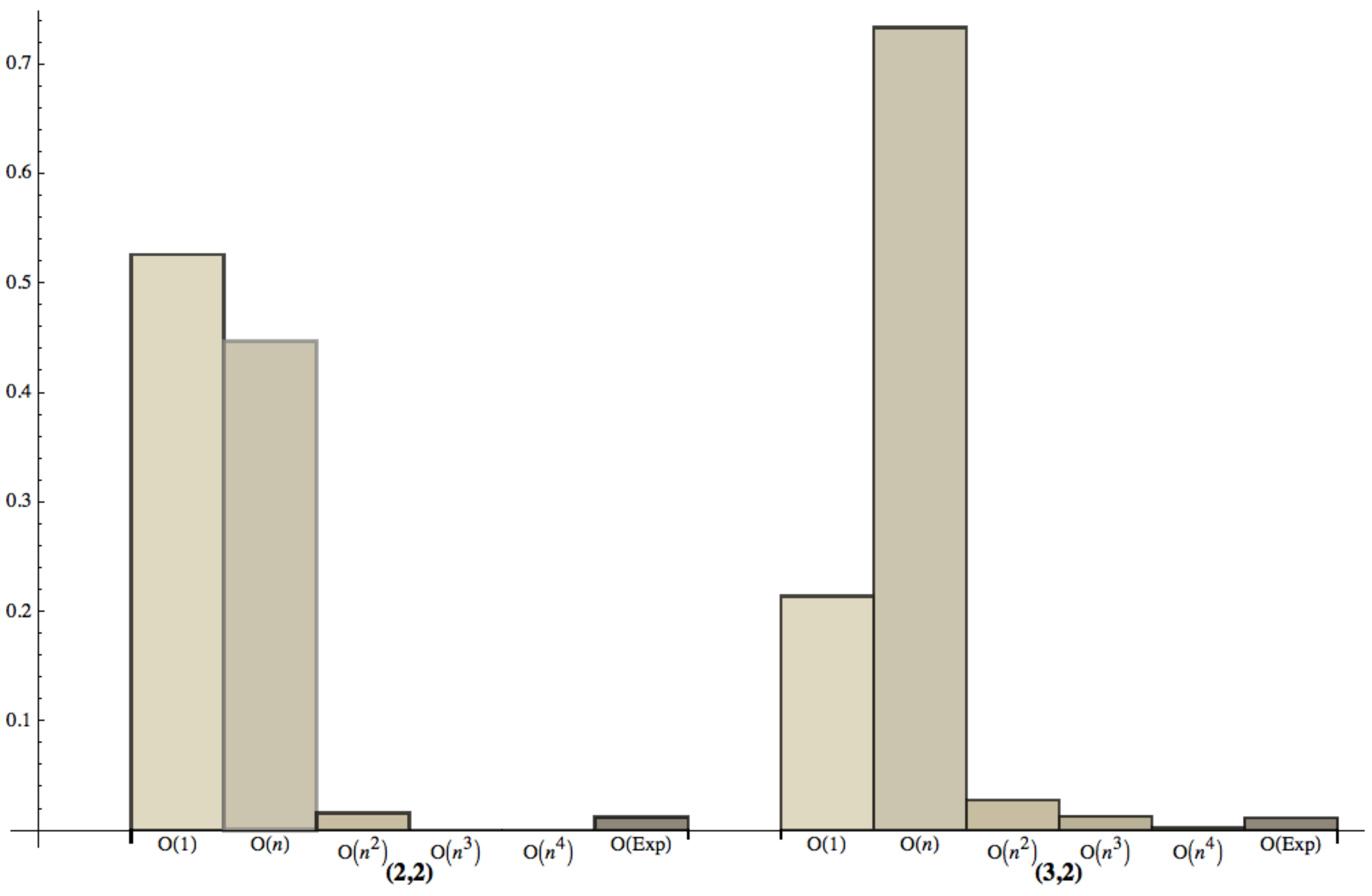}
  \caption{Time complexity distributions of (2,2) (left) and (3,2) (right).}
  \label{timecomp1}
\end{figure}

We have far to few data to possibly speak of a prior in the distributions of our TMs over these complexity classes. However, we do remark the following. In the following table we see the fraction per complexity class of the non-constant TMs for each space. Even though for (4,2) we do not at all work with a representative sampling still there is some similarity in the fractions. Most notably within one TM space, the ratio of one complexity class to another is in the same order of magnitude as the same ratio in one of the other spaces. Notwithstanding this being a far cry from a prior, we do find it worth\footnote{Although we have very few data points we could still audaciously calculate the Pearson coefficient correlations between the classes that are inhabited within one of the spaces. Among (2,2), (3,2) and (4,2) the Pearson coefficients are: 0.999737, 0.999897 and 0.999645.} while mentioning.


\[
\begin{array}{|c|c|c|c|}
\hline
\textrm{pr} & \textrm{(2,2)} & \textrm{(3,2)} & \textrm{(4,2)} \\
\hline
\textrm{O(n)} & 0.941667 & 0.932911 & 0.925167 \\
\hline
\left.\textrm{O(}n^2\right) & \ \ 0.0333333 \ \ \ & \ \ 0.0346627 \ \ \ & \ \ 0.0462362 \ \ \ \\
\hline
\left.\textrm{O(}n^3\right) & 0 & 0.0160268 & 0.0137579 \\
\hline
\left.\textrm{O(}n^4\right) & 0 & 0.0022363 & 0.00309552 \\
\hline
\ \ \ \textrm{O(Exp)}\ \ \ \  & 0.025 & 0.0141633 & 0.0117433\\
\hline
\end{array}
\]

\subsection{Quantifying the linear speed-up factor}
\label{linear}

For obvious reasons all functions computed in (2,2) are computed in (3,2). The most salient feature in the comparison of the (2,2) and (3,2) spaces is the prominent slowdown indicated by both the arithmetic and the harmonic averages. The space (3,2) spans a larger number of runtime classes. Figures \ref{comp1} and \ref{comp2} are examples of two functions computed in both spaces in a side by side comparison with the information of the function computed in (3,2) on the left side and the function computed by (2,2) on the right side. In \cite{JoostenSolerZenilDemonstration} a full overview of such side by side comparison is published. Notice that the numbering scheme of the functions indicated by the letter \emph{f} followed by a number may not be the same because they occur in different order in each of the (2,2) and (3,2) spaces but they are presented side by side  for comparison with the corresponding function number in each space. 

\begin{figure}
  \centering{}
  \includegraphics[width=9cm]{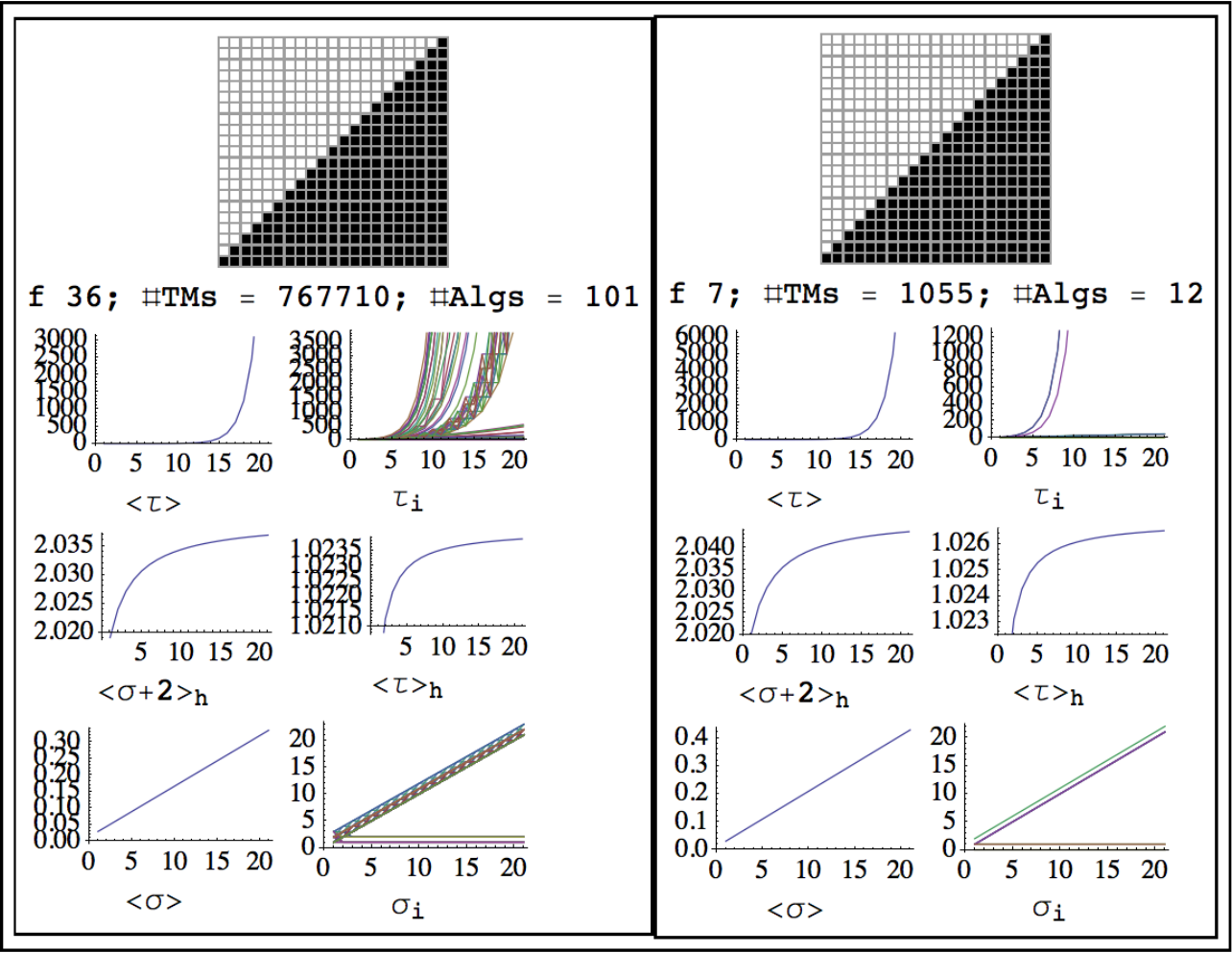}
  \caption{Side by side comparison of an example computation of a function in (2,2) and (3,2) (the identity function).}
  \label{comp1}
\end{figure}

\begin{figure}
  \centering
  \includegraphics[width=8.8cm]{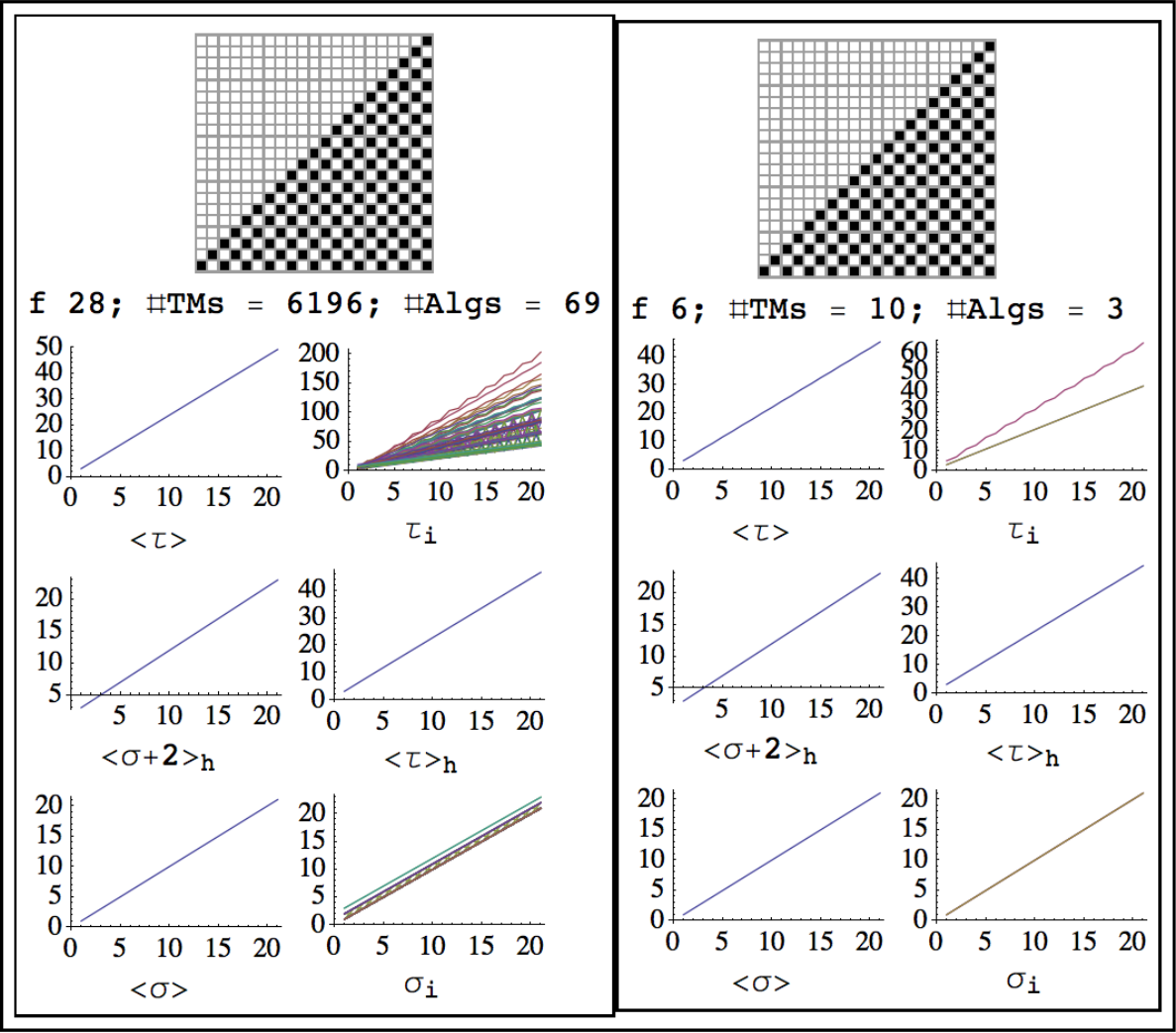}
    \caption{Side by side comparison of the computation of a function in (2,2) and (3,2).}
  \label{comp2}
\end{figure}

One important calculation experimentally relating descriptional (program-size) complexity and (time resources) computational complexity is the comparison of maximum of the average runtimes on inputs 0,$\ldots$,20, and the estimation of the speed-ups and slowdowns factors found in (3,2) with respect to (2,2). 

It turns out that 19 functions out of the 74 computed in (2,2) and (3,2) had at least one fastest computing algorithm in (3,2). That is a fraction of 0.256 of the 74 functions in (2,2). A further inspection reveals that among the 3\,414 algorithms in (3,2), computing one of the functions in (2,2), only 122 were faster. If we supposed that ``chances'' of speed-up versus slow-down on the level of algorithms were fifty-fifty, then the probability that we observed at most 122 instantiations of speed-up would be in the order of $10^{-108}$. Thus we can safely state that the phenomena of slow-down at the level of algorithms is significant.

Figure \ref{probdist} shows the scarceness of the speed-up and the magnitudes of such probabilities. Figures \ref{speedup} quantify the linear factors of speed-up showing the average and maximum. The typical average speed-up was 1.23 times faster for an algorithm found when there was a faster algorithm in (3,2) computing a function in (2,2).

\begin{figure}
  \centering
  \includegraphics[width=8cm]{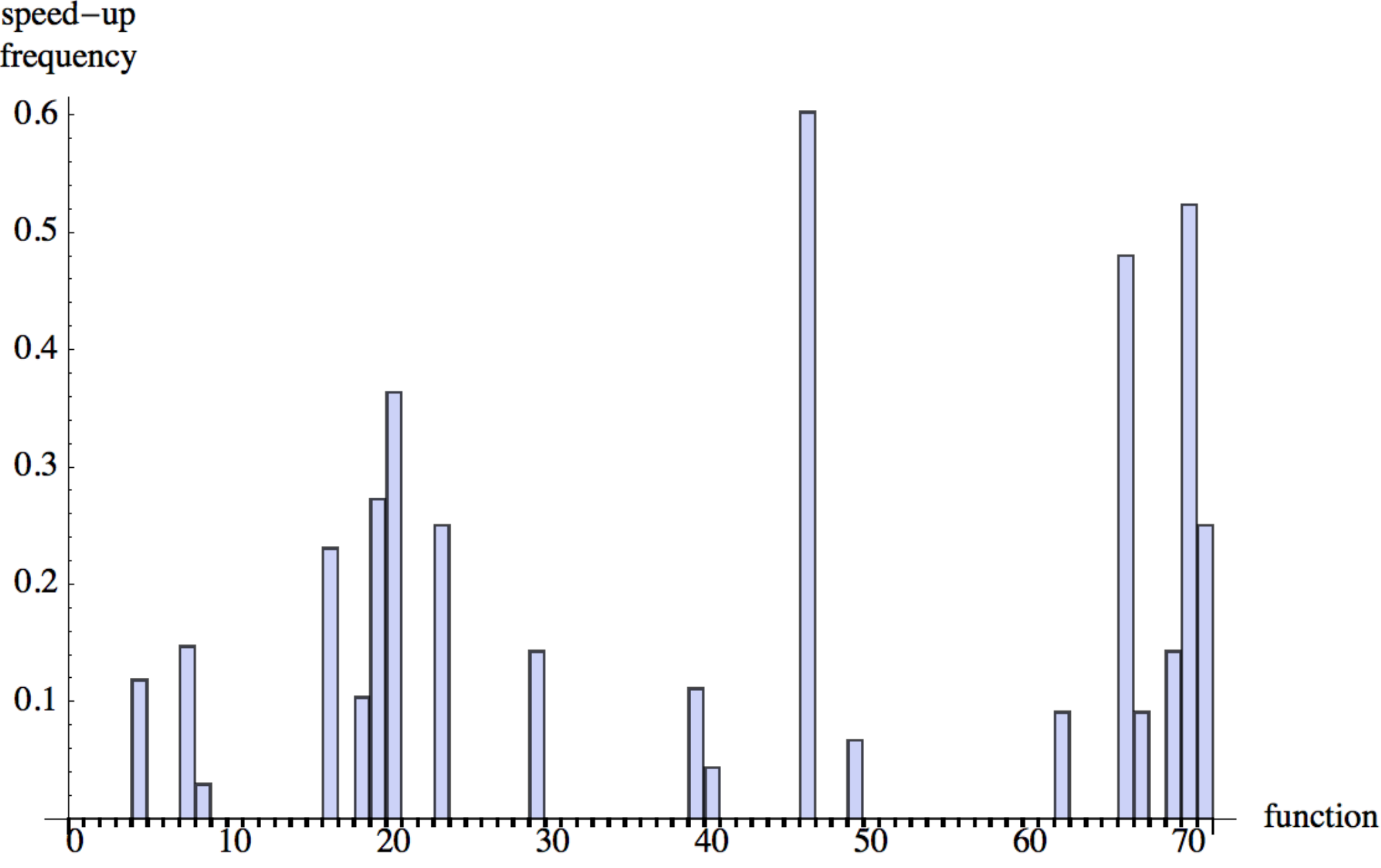}
    \caption{Distribution of speed-up probabilities per function. Interpreted as the probability of picking an algorithm in (3,2) computing faster an function in (2,2).}
  \label{probdist}
\end{figure}

In contrast, slowdown was generalized, with no speed-up for 0.743 of the functions. Slowdown was not only the rule but the significance of the slowdown was much larger than the scarce speed-up phenomenon. The average algorithm in (3,2) took 2\,379.75 longer and the maximum slowdown was of the order of $1.19837\times10{^6}$ times slower than the slowest algorithm computing the same function in (2,2).

\begin{figure}
  \centering
  \includegraphics[width=6cm]{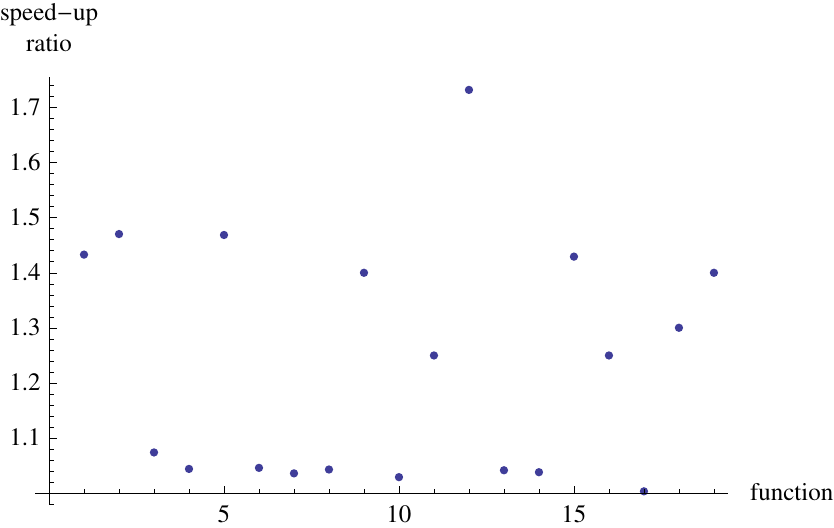}
  \includegraphics[width=6cm]{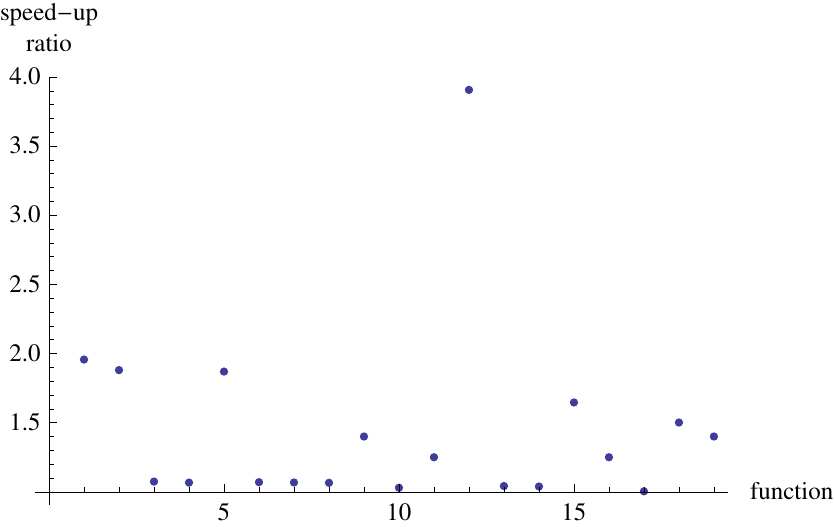}
    \caption{Speed up significance: on the left average and on the right maximum speed-ups.}
  \label{speedup}
\end{figure}

As mentioned before there is also no essential speed-up in the space (4,2) compared to (3,2) and only linear speed-up was witnessed at times. But again, slow-down was the rule. Thus, (4,2) confirmed the trend between (2,2) and (3,2), that is that linear speed up is scarce yet present, three functions (0.0069) sampled from (3,2) had faster algorithms in (4,2) that in average took from 2.5 to 3 times less time to compute the same function, see Figure~\ref{Figure:SpeedUpFactors32Versus42}. 

\begin{figure}
  \centering
  \includegraphics[width=8cm]{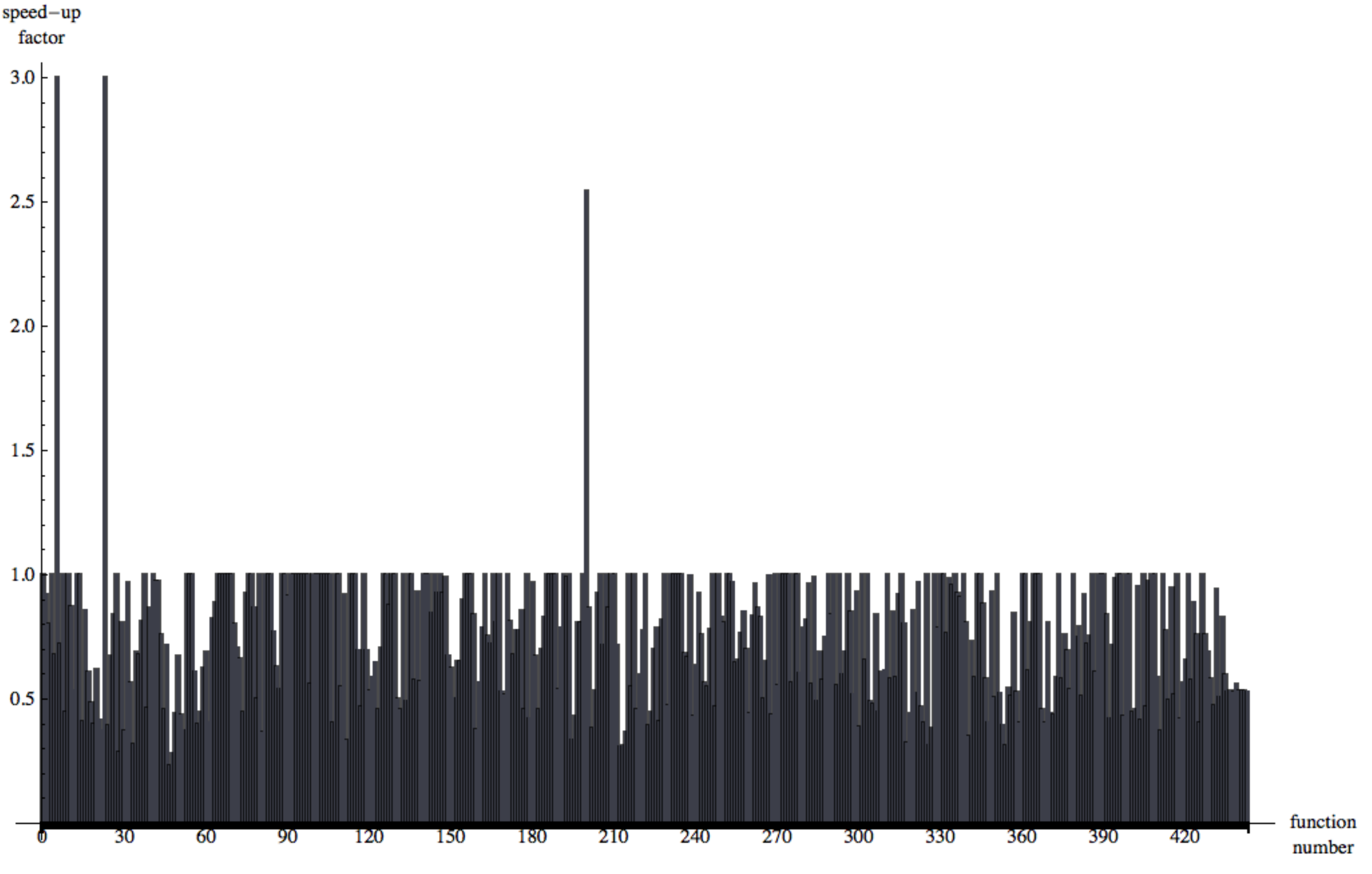}
    \caption{Distribution of average speed-up factors among all selected 429 functions computed in (3,2) and (4,2).}
  \label{probdist}
\end{figure}\label{Figure:SpeedUpFactors32Versus42}

\section{Summarizing}
In this enterprise we have in approximation determined and computed all functions that are computable on Turing Machines with two tape symbols and at most three states. The peculiarities of these functions have been investigated and both spaces of Turing Machines are compared to each other. The study and comparison was then extended to a sample of Turing machines with four states and two symbols.

The most prominent findings have been:
\begin{itemize}
\item On average, functions tend to be computed slower when more states are added

\item
The halting distributions (per number of computation steps, the number of pairs (TM, input) that halt in this amount of steps) drop quickly and exhibit clear phase-transitions.

\item
The behavior of a TM on a very small initial segment of input values determines the full behavior of that TM on all input values

\end{itemize}
Apart from these main findings the paper addresses additional features of TM spaces, like the definable sets, the possible type of computations, clustering phenomena, the way in which the number of states is manifested in the type of computations and computable functions, and more related questions.

Notwithstanding the fact that the realm of the Turing Machine spaces seems simple and for small number of states no computationally universal phenomena occur yet, the authors were surprised and struck by the overwhelmingly rich and intriguing structures that were found in what they became to call \emph{the micro-cosmos of small Turing machines}.

\section*{Acknowledgements}
The initial motivation and first approximation of this project was developed during the NKS Summer School 2009 held at the Istituto di Scienza e Tecnologie dell'Informazione, CNR in Pisa, Italy. We wish to thank Stephen Wolfram for interesting suggestions and guiding questions. Furthermore, we wish to thank the CICA center and its staff for providing access to their supercomputing resources and their excellent support. 

\end{document}